\documentclass{theoretics}
\title{The Converse of the Real Orthogonal Holant Theorem}

\ThCSauthor{Ben Young}{benyoung@cs.wisc.edu}[0000-0003-1921-7253]
\ThCSaffil{Department of Computer Sciences, University of Wisconsin-Madison, USA}
\ThCSthanks{An abridged version of this work appeared in the Proceedings of the 52nd International Colloquium on Automata, Languages, and Programming (ICALP 2025) \cite{orthogonal}}
\ThCSshortnames{B.\ Young}  
\ThCSshorttitle{The Converse of the Real Orthogonal Holant Theorem}
\ThCSyear{2026}
\ThCSarticlenum{15}
\ThCSreceived{Sep 17, 2025}
\ThCSaccepted{Aug 9, 2026}
\ThCSrevised{Jun 23, 2026}
\ThCSpublished{Oct 8, 2026}
\ThCSdoicreatedtrue
\ThCSkeywords{Holant, Counting Indistinguishability, Orthogonal Transformation, Odeco}

\usepackage{subcaption}
\usepackage{tikz}
\usepackage{tkz-graph}
\usepackage{mathdots}
\usepackage{dsfont}
\usetikzlibrary{shapes}

\DeclareMathOperator{\spn}{span}
\DeclareMathOperator{\tr}{tr}
\DeclareMathOperator{\rr}{\mathbb{R}}
\DeclareMathOperator{\cc}{\mathbb{C}}
\DeclareMathOperator{\arity}{arity}
\DeclareMathOperator{\vk}{\mathbf{K}}
\DeclareMathOperator{\vl}{\mathbf{L}}

\DeclareMathOperator{\vx}{\mathbf{x}}
\DeclareMathOperator{\vy}{\mathbf{y}}
\DeclareMathOperator{\vz}{\mathbf{z}}
\DeclareMathOperator{\va}{\mathbf{a}}
\DeclareMathOperator{\vb}{\mathbf{b}}

\DeclareMathOperator{\vr}{\mathbf{R}}
\DeclareMathOperator{\vrr}{\mathbf{r}}
\DeclareMathOperator{\vc}{\mathbf{C}}
\DeclareMathOperator{\vcc}{\mathbf{c}}
\DeclareMathOperator{\vj}{\mathbf{J}}
\DeclareMathOperator{\vjj}{\mathbf{j}}
\DeclareMathOperator{\ic}{\mathcal{I}}
\DeclareMathOperator{\fc}{\mathcal{F}}
\DeclareMathOperator{\gc}{\mathcal{G}}
\DeclareMathOperator{\gk}{\mathfrak{G}_{\mathcal{F}}}
\DeclareMathOperator{\qk}{\mathfrak{Q}_{\mathcal{F}}}
\DeclareMathOperator{\stab}{Stab}
\newcommand{\tcwd}[1]{\left\langle #1 \right\rangle_{+, \circ, \otimes ,\top}}
\newcommand{\tcwdn}[1]{\left\langle #1 \right\rangle_{\circ, \otimes ,\top}}
\DeclareMathOperator{\holant}{Holant}
\DeclareMathOperator{\plholant}{Pl-Holant}
\DeclareMathOperator{\csp}{\#CSP}

\DeclareMathOperator{\eq}{\mathcal{EQ}}
\DeclareMathOperator{\geneq}{\mathcal{GEQ}}

\makeatletter
\newcommand\IfRestateTF{%
  \ifx\label\thmt@gobble@label 
    \expandafter\@firstoftwo
  \else
    \expandafter\@secondoftwo
  \fi
}
\makeatother
\newcommand{\RestateRemark}{\IfRestateTF{{\normalfont\bfseries (Restated) }}{}}

\begin{document}
\maketitle

\begin{abstract}
The Holant theorem is a powerful tool for studying the computational complexity of
counting problems. Due to the great expressiveness of the Holant framework,
a converse to the Holant theorem would itself be a very powerful 
\emph{counting indistinguishability theorem}. The most general converse does not hold, but we prove
the following, still highly general, version:
if any two sets of real-valued signatures are Holant-indistinguishable, then they are equivalent 
up to an orthogonal transformation. This resolves a partially open conjecture of Xia (2010).
Consequences of this theorem include the well-known result that homomorphism counts from all graphs
determine a graph up to isomorphism, the classical sufficient condition for simultaneous orthogonal
similarity of sets of real matrices,
and a combinatorial characterization of simultaneously orthogonally decomposable (odeco) sets of 
symmetric tensors.
\end{abstract}


\section{Introduction}
\paragraph{Holant problems.}

Holant problems were introduced by Cai, Lu, and Xia \cite{cai_computational_2011} as a highly
expressive framework for studying the computational complexity of counting problems. 
The problem $\holant(\fc)$ is defined by a set $\fc$ of \emph{signatures}, 
where a signature $F$ of \emph{arity}
$n$ on \emph{domain} $[q] := \{0,1,\ldots,q-1\}$ is a tensor in $(\cc^q)^{\otimes n}$, or equivalently
a function $[q]^n \to \cc$.
Given a \emph{signature grid} $\Omega$ -- a multigraph in which every degree-$n$ vertex is assigned a
$n$-ary signature from $\fc$ --
the problem is to compute the \emph{Holant value} of $\Omega$, which is the value
of the contraction of $\Omega$ as a tensor network (see \autoref{sec:holant} for formal definitions). 
For various $\fc$, $\holant(\fc)$ captures a wide variety of natural counting problems on
graphs, including counting partial or perfect matchings, graph homomorphisms, 
proper vertex or edge-colorings, or Eulerian orientations.
Major complexity dichotomies classifying $\holant(\fc)$ as either 
polynomial-time tractable or \#P-hard,
depending on $\fc$, have been proved for various combinations of restrictions on $\fc$ -- for example,
requiring that the signatures in $\fc$ be real- or nonnegative-real-valued, symmetric
(invariant under reordering of their inputs), or on the Boolean domain $q=2$
\cite{huang_2016_dichotomy,cai_complete_2016,cai2015holant,lin_complexity_2018,shao}.

Holant problems were motivated by Valiant's technique of \emph{holographic
transformations} \cite{valiant}. In particular, Valiant's \emph{Holant theorem} (\autoref{thm:holant}
below) states roughly that two any signature sets $\fc$ and $\gc$ that are 
equivalent up to a certain 
linear transformation are \emph{Holant-indistinguishable}, meaning that each signature grid $\Omega$
has the same Holant value whether its vertices are assigned signatures from $\fc$ or from 
$\gc$.
Many problems which do not otherwise appear tractable are in fact tractable under a Holographic
transformation to a known tractable problem \cite{valiant_2006_accidental}.
Xia \cite{xia} conjectured the converse of the Holant theorem: if $\fc$ and
$\gc$ are Holant-indistinguishable, then they are equivalent up to linear transformation.
Xia's general conjecture is false \cite{cai_complete_2016}, but one case highlighted
by Xia was left open. This paper proves that case, which is as follows.
\begin{theorem}[\autoref{thm:result}, informal]
    Let $\fc$ and $\gc$ be sets of real-valued signatures. Then $\fc$ and $\gc$ are equivalent under
    a real orthogonal transformation if and only if $\fc$ and $\gc$ are Holant-indistinguishable.
\end{theorem}
Later, the techniques developed in a previous version of
the present work 
\cite{orthogonal} (including the direct sum and subdomain-restriction ideas in
Lemmas \ref{lem:offdiagblock} and \ref{lem:restriction} below) were used
to show that Xia's general conjecture holds for any $\fc$ and
$\gc$ which are \emph{quantum-nonvanishing} \cite{bipartite}.
The $\fc$ and $\gc$ in the context of
\autoref{thm:result} are quantum-nonvanishing, so a slightly weaker version of 
\autoref{thm:result} in which the obtained orthogonal matrix is complex 
also follows from this more general converse. 

\paragraph{Vertex and edge coloring models.}
This work uses and generalizes ideas from the theory of vertex coloring models and edge coloring models,
two well-studied classes of Holant problems.
De la Harpe and Jones \cite{jones} defined vertex and edge coloring models as extensions of 
statistical mechanics models (e.g. the Ising model), calling them ``spin models'' and ``vertex models'',
respectively. A vertex coloring model (also called a spin system) is defined by a graph $X$ with edge and possibly vertex
weights. Given an input graph $K$, one aims to compute the \emph{partition function}, the number of 
(weighted) graph homomorphisms from $K$ to $X$.
An edge coloring model is defined by a set $\fc$ of symmetric signatures containing exactly
one signature of each arity, and the problem of computing its partition function is equivalent to 
$\holant(\fc)$ (this restriction on $\fc$ ensures that
edge coloring models take ordinary graphs, rather than signature grids, as input).

One thread of prior work on vertex and edge coloring models 
characterizes which scalar-valued functions on graphs
are expressible as vertex coloring models \cite{freedman_reflection, schrijver}
or as edge coloring models 
\cite{szegedy_edge_2007, schrijver_graph_2008, draisma_characterizing_2012, regts_characterization_2013}.
Another, related, line of works compute the rank of \emph{connection matrices}
for vertex coloring models \cite{lovasz} and edge coloring models 
\cite{regts_rank_2012,draisma_tensor_2013}.
See \cite{regts} for an overview of many of the above results. 
Following Freedman, Lov\'{a}sz, and Schrijver \cite{freedman_reflection}, these works use
(labeled) \emph{quantum graphs}, algebras of formal linear combinations of graphs 
equipped with labeled vertices or ``half edges'' incident to a single vertex,
special cases of our \emph{quantum gadgets} below.
Many of these works also apply techniques from invariant theory, either of the symmetric group
in the case of vertex coloring models \cite{schrijver}, or, as in this work, of the
orthogonal group $O(q)$ in the case of edge-coloring models \cite{szegedy_edge_2007, schrijver_graph_2008, draisma_characterizing_2012, regts_rank_2012, draisma_tensor_2013, regts_characterization_2013}.

\paragraph{Counting Indistinguishability Theorems.}
\autoref{thm:result} is a very general and powerful algebraic 
\emph{counting indistinguishability theorem}. 
Such a theorem proves that two signatures, or sets of signatures, are indistinguishable
as parameters for a counting problem if and only if they are equivalent under an algebraic
transformation. These theorems exist for both vertex and edge coloring models, as well as other
counting problems \cite{dvorak_recognizing_2010, sherali-adams, dell, grohe_homomorphism_2022, lasserre}.
Since Holant captures a wide variety of counting problems,
many such theorems are special cases
of \autoref{thm:result} (see \autoref{sec:corollaries}).
If the counting problem in question is a generalized vertex coloring model ($\holant(\fc \cup \eq)$
for some $\fc$), then the algebraic transformation is isomorphism, and if the counting problem is,
as in this work, a generalized edge coloring model, then it is orthogonal.
The first counting indistinguishability theorem, proved by Lov\'{a}sz \cite{lovasz_operations}, 
states that two graphs are
isomorphic if and only if they admit the same number of homomorphisms from all graphs. Much later,
Lov\'{a}sz \cite{lovasz} extended this theorem to vertex coloring models with nonnegative real 
weights, followed by extensions to complex edge weights by Schrijver 
\cite{schrijver}, and to weights from any field of characteristic zero by Cai and Govorov 
\cite{cai-lovasz}. Young \cite{young2022equality} extended Cai and Govorov's proof to \#CSP, or
$\holant(\fc \cup \eq)$ for any $\fc$. 

For edge coloring models, Schrijver \cite{schrijver_graph_2008} showed that 
$\fc$ and $\gc$ define indistinguishable real edge coloring models if and only if
$\fc$ and $\gc$ are equivalent under a real orthogonal transformation. 
This is a special case of our \autoref{thm:result}.
Schrijver's proof exploits the specific nature of edge coloring models -- that $\fc$ and $\gc$ consist
of symmetric signatures and exactly one signature per arity -- to transform input graphs into polynomials
expressible in variables $y_1,\ldots,y_q$ (for $\fc$ and $\gc$ on domain $[q]$), 
where a monomial with variable multiset $\{y_{i_1},\ldots, y_{i_n}\}$ corresponds to the
$\{i_1,\ldots,i_n\}$-entry of the unique $n$-ary signature in $\fc$.
Another form of this result (in which the orthogonal matrix is allowed to be complex) 
follows from Regts' proof of 
\cite[Lemma 5]{regts_characterization_2013}, which similarly encodes $\fc$ and $\gc$ as polynomials.

Man\v{c}inska and Roberson \cite{planar} 
showed that two graphs are \emph{quantum isomorphic} -- a relaxation of isomorphism originally
defined in terms of quantum strategies for a nonlocal game \cite{nonsignalling}
-- if and only if they admit the same number of homomorphisms from all planar graphs.
Cai and Young \cite{cai_planar_2023} extended this result to
planar $\csp$ (or $\plholant(\fc \cup \eq)$, where $\plholant$ restricts to planar signature grids). 

The theory of counting indistinguishability is closely related to the computational complexity
of counting problems, especially graph homomorphisms. The quantum graph algebras of Freedman,
Lovász, and Schrijver \cite{freedman_reflection} -- which, along with their generalizations 
such as, in this work, quantum gadgets, have played central roles in many indistinguishability
results \cite{lovasz,schrijver,planar,grohe_homomorphism_2022,cai_planar_2023}
-- have also been employed to study the parameterized complexity of subgraph counting
\cite{curticapean2017homomorphisms,curticapean2024count} and the complexity of counting
homomorphisms from planar graphs \cite{cai_planar_2026}. The complexity of homomorphism
indistinguishability itself -- that is, of the problem of determining, for a fixed graph
class $\mathfrak{G}$, whether input graphs $X$ and $Y$ admit the same number of homomorphisms
from every graph in $\mathfrak{G}$, has also seen recent study \cite{seppelt2024algorithmic}.
Simultaneously, the complexity of the problem of determining the existence of invertible, 
including orthogonal, transformations between two tensors, called \emph{tensor isomorphism}
in this context, has also been investigated \cite{complexity_isomorphism1,complexity_isomorphism3}. The connection between indistinguishability and tensor isomorphism developed in this 
work is a first step towards developing connections between 
these previously disparate areas of complexity theory.

\paragraph{Odeco signature sets}
A real-valued symmetric signature (tensor) is \emph{orthogonally decomposable}, 
or \emph{odeco} \cite{robeva}, if it is orthogonally transformable to a signature in $\geneq$, the set
of \emph{generalized equality} signatures, which take nonzero values only when all of their inputs
are equal. Hence odeco tensors generalize diagonalizable matrices.
Call a set $\fc$ of signatures odeco if the signatures are simultaneously odeco (there
is a single orthogonal transformation mapping $\fc$ into $\geneq$).
In counting complexity, if $\fc$ is odeco, then
$\holant(\fc)$ is polynomial-time tractable, as $\fc$ maps into $\geneq$, a trivially tractable set,
under an orthogonal holographic transformation.
Indeed, the tractability of \emph{Fibonacci} signature sets \cite{cai2008holographic}
can, with one exception, be explained by such 
sets being simultaneously odeco (see e.g. \cite[Section 2.2]{cai_chen_2017}). 
Fibonacci sets constitute almost all nontrivial tractable cases of $\holant^*$ problems
(an important variant of Holant in which all unary signatures are assumed 
present) for symmetric signatures on the Boolean domain \cite{cai_computational_2011}.
Simultaneously odeco
sets provide a natural starting point for extending Fibonacci signatures to higher
domains \cite{liu_2024,domain3}, where no full complexity dichotomy for $\holant^*$ is known.

Boralevi, Draisma, Horobeţ, and Robeva \cite{boralevi_orthogonal_2017}, resolving a conjecture of
Robeva \cite{robeva}, showed using techniques from algebraic geometry that a single tensor $F$ is 
odeco if and only if the signature of 
a certain $F$-gadget is symmetric. Using \autoref{thm:result}, we in \autoref{thm:odeco} extend this 
characterization to sets of simultaneously odeco signatures: $\fc$ is odeco if and only if every connected 
$\fc$-gadget has a symmetric signature. The latter condition is 
equivalent to the symmetry of the signatures of a set of small gadgets constructed from 
every pair of 
signatures in $\fc$. Therefore, if $\fc$ is finite, our characterization yields a simple
$O(q^{2n-2} |\fc|^2)$-time algorithm (for $\fc$ on domain $[q]$ with maximum arity $n$) for deciding
whether $\fc$ is odeco.
Our characterization also deepens the connection between Fibonacci and odeco 
signatures, as the original proof of tractability of any Fibonacci signature set $\fc$
\cite{cai2008holographic} relied on the fact that every connected $\fc$-gadget has a signature which is
itself Fibonacci (in particular, is symmetric). 
One can view the (iii) $\implies$ (ii) result in \autoref{thm:odeco} as a general-domain version
of this proof.

\paragraph{Overview.}
We introduce the necessary preliminaries and formally state the main theorem, 
\autoref{thm:result}, in 
\autoref{sec:preliminaries}. We prove \autoref{thm:result} in Sections
\ref{sec:proof} and \ref{sec:duality}. The latter \autoref{sec:duality} is mostly algebraic in nature.
It applies an invariant-theoretic result of Schrijver \cite{schrijver_tensor_2008} to prove 
an algebraic characterization (\autoref{thm:intertwiner_gadget}) of quantum gadgets,
which we call \emph{duality}.
The proof of this characterization follows and generalizes a proof of a similar result of 
Regts \cite{regts_rank_2012}
for edge coloring models. By unifying the perspective of Regts with that of
Man\v{c}inska and Roberson and Cai and Young \cite{planar, cai_planar_2023, young2022equality}
we find our proof analogous to proofs of similar results in the latter line of work
(see \autoref{rem:duality}). 
The former \autoref{sec:proof} is mostly combinatorial in nature, and uses the duality theorem of the 
\autoref{sec:duality} as a black box. While this duality is analogous to those used by
Man\v{c}inska and Roberson and Cai and Young, the proofs of their ensuing
counting indistinguishability theorems use the discrete \emph{orbits}
of the action of the symmetric or quantum symmetric group on the domain set 
$[q]$, which do not exist for the continuous group $O(q)$ of orthogonal transformations in
our setting. Instead, we apply a novel method: induction on the domain size $q$. 
We show that, if $\fc$ and $\gc$ contain a binary signature
represented by a nontrivial (i.e., not a multiple of $I$) diagonal matrix $D$, 
then we can separate $[q]$ into smaller subdomains and, for each such subdomain 
$X \subset [q]$, inductively obtain an orthogonal transformation between the restrictions
$\fc|_X$ and $\gc|_X$ of $\fc$ and $\gc$ to $X$,
then combine these into a full orthogonal transformation between $\fc$ and $\gc$. 
We then use duality to nonconstructively obtain this matrix $D$.

In \autoref{sec:corollaries}, we show that \autoref{thm:result} 
encompasses a wide range of existing counting indistinguishability theorems,
and yields some novel variations of these theorems. We also prove our combinatorial characterization of
odeco signature sets.
Finally, in \autoref{sec:variations}, we conjecture a variation of \autoref{thm:result}
extending the results 
of Man\v{c}inska and Roberson \cite{planar} and Cai and Young \cite{cai_planar_2023} 
to planar-Holant-indistinguishability and quantum orthogonal transformations.

\section{Preliminaries, Background, and the Main Theorem}
\label{sec:preliminaries}
\subsection{Holant Problems}
\label{sec:holant}
A \emph{multigraph} $K = (V(K),E(K))$ is defined by a finite
set $V(K)$ of vertices and a multiset $E(K)$ of edges, 
which are one- or two-element subsets of 
$V(K)$. That is, edges are undirected and loops are allowed. A \emph{graph} $K$ has no
loops and at most one edge between every pair of vertices.

A \emph{signature} $F$ of finite \emph{arity} $n \geq 0$ on finite \emph{domain} $V(F)$ is a 
function $V(F)^n \to \rr$.
We will often take $V(F) = [q] := \{0,1,\ldots,q-1\}$, in which case we also view $F$ as a
tensor in $(\rr^q)^{\otimes n}$. For $\vx = (x_1,\ldots,x_n) \in V(F)^n$, abbreviate
$F_{\vx} := F(x_1,\ldots,x_n) \in \rr$. Signature $F$ is \emph{symmetric} if 
its value depends only on the multiset of inputs, not on their order.
Any time we consider a set $\fc$ of signatures,
we assume that all signatures in $\fc$ have the same domain, denoted $V(\fc)$. 
\begin{definition}[$T(\rr^q)$]
    Following Schrijver \cite{schrijver_tensor_2008}, define
    \[
        T(\rr^q) := \bigcup_{n \geq 0} (\rr^q)^{\otimes n}
    \]
    to be the set of all real-valued signatures on domain $[q]$.
\end{definition}

In the context of a signature set $\fc$, a
\emph{signature grid} (or \emph{$\fc$-grid}) $\Omega$ consists of an underlying multigraph with 
vertex set $V$ and edge set $E$, and an assignment of an $\deg(v)$-ary signature $F_v \in \fc$ 
to each $v \in V$, along with an ordering of the edges $\delta(v) \subset E$ incident to $v$ 
to serve as the
$\deg(v)$ input variables to $F$. That is, there is an ordering $e_1,\ldots,e_{\deg(v)}$ of $v$'s 
incident edges such that, if $\sigma: E \to V(\fc)$ is an assignment of a value in $V(\fc)$ to each
edge of $\Omega$, then $F_v$ evaluates to 
$F_v(\sigma|_{\delta(v)}) := F_v(\sigma(e_1),\ldots,\sigma(e_{\deg(v)}))$.
Somewhat unusually, we also allow $E$ to contain \emph{vertexless loops}, edges whose two ends are 
connected to each other, with no incident vertices.
The problem $\holant(\fc)$ is defined as follows: given an $\fc$-grid $\Omega$
with vertex set $V$ and edge set $E$, compute the \emph{Holant value}
\begin{equation}
    \holant_\Omega(\fc) := \sum_{\sigma: E \to V(\fc)} \prod_{v \in V} F_v(\sigma|_{\delta(v)}).
    \label{eq:holant}
\end{equation}
When $\fc$ is clear from context, we abbreviate $\holant_\Omega(\fc)$ as 
$\holant_\Omega$. Each vertexless loop of $\Omega$ contributes a global factor $|V(\fc)|$ to
$\holant_\Omega$. More generally, the Holant value of a disconnected signature grid
is the product of the Holant values of its connected components.
For signature sets $\fc$ and $\fc'$, define $\holant(\fc \mid \fc')$ as the Holant problem 
whose input is a bipartite $(\fc\sqcup\fc')$-grid
$\Omega$ (where $\sqcup$ denotes disjoint union), with the vertices in the two bipartitions 
assigned signatures in $\fc$ and $\fc'$, respectively.

For example, if $\fc$ is on the Boolean domain $\{0,1\}$ and consists of, for each arity $n$, a symmetric
signature which evaluates to 1 on input strings of Hamming weight 1 and 0 on all other input strings,
then $\holant_{\Omega}(\fc)$ equals the number of perfect matchings in the multigraph underlying
$\Omega$. For another example, consider counting graph homomorphisms from a 
multigraph $K$ to a graph $X$, where a graph homomorphism from $K$ to $X$ is a map
$\varphi: V(K) \to V(X)$ such that $uv \in E(K) \implies \varphi(u)\varphi(v) \in V(X)$.
Let $A_X \in \{0,1\}^{q \times q}$ be the 
adjacency matrix of $q$-vertex graph $X$, and define the set
$\eq = \{=_n \mid n \geq 1\}$ of \emph{equality}
signatures on domain $[q]$, where $=_n(x_1,\ldots,x_n)$ is 1 if $x_1 = \ldots = x_n$, 
and is 0 otherwise. Consider a $\holant(A_X \mid \eq)$ grid $\Omega$. 
We can think of any edge assignment
$\sigma$ with a nonzero contribution to the Holant value as a map from vertices in $\Omega$ assigned
$\eq$ signatures to values in $[q]$, or equivalently vertices of $X$. This map must send the two 
$\eq$ vertices
adjacent to every $A_X$ vertex to an edge in $A_X$, so, if $K_\Omega$ is the multigraph 
resulting from ignoring
(i.e. treating as edges) the degree-two vertices assigned $A_X$ in the underlying 
multigraph of $\Omega$,
then $\holant_{\Omega}(A_X \mid \eq)$ equals the number of homomorphisms
from $K_\Omega$ to $X$. See \autoref{fig:hom}.
\begin{figure}[ht!]
    \centering
\includegraphics[width=0.8\textwidth]{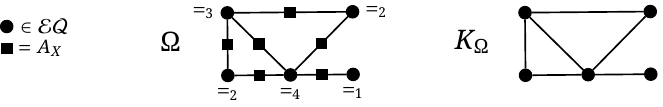}
    \caption{An $(A_X \mid \eq)$-grid $\Omega$ and the multigraph $K_\Omega$ such that $\holant_\Omega(A_X \mid \eq)$
    counts the number of homomorphisms from $K_\Omega$ to $X$.}
    \label{fig:hom}
\end{figure}

\subsection{Gadgets and Signature Matrices}
Instead of viewing a signature $F$ as a tensor in $(\rr^q)^{\otimes n}$ or function in $\rr^{[q]^n}$,
we can partition its inputs in two to view it naturally as a matrix.
\begin{definition}[$F^{m,d},f$]
    \label{def:sig_matrix}
    For $F \in (\rr^q)^{\otimes n}$ and $m,d \geq 0$ with $m+d=n$, define the 
    $(m,d)$-\emph{signature matrix}, or \emph{flattening}, $F^{m,d} \in \rr^{q^m\times q^d}$ of $F$
    by, for $\vx \in [q]^m$ and $\vy \in [q]^d$,
    \[
        (F^{m,d})_{\vx,\vy} = F(x_0,\ldots,x_{m-1},y_{d-1},\ldots,y_0),
    \]
    where we use the standard isomorphism $[q^n] \cong [q]^n$ to index $F^{m,d}$.

    Abbreviate $f = F^{n,0} \in \rr^{q^n}$ -- the (column) \emph{signature vector} of $F$.
\end{definition}
We will often identify binary signatures in $(\rr^q)^{\otimes 2}$
with their $1,1$ signature matrices in $\rr^{q \times q}$.
\begin{definition}[Gadget, $\gk,\gk(m,d)$]
    For signature set $\fc$, an $\fc$-\emph{gadget} $\vk$ is identical to an $\fc$-grid,
    but with an enhanced edge set $E(\vk) = E'(\vk) \cup L(\vk) \cup R(\vk)$, 
    where $E'(\vk)$ is the
    usual set of $\fc$-grid \emph{internal} edges and $L(\vk)$ and $R(\vk)$ are ordered
    lists of left and right \emph{dangling edges}, each specified by either a single endpoint 
    in $V(\vk)$, or paired with another element of $L(\vk)$ or $R(\vk)$ as
    a two-sided dangling edge -- a \emph{wire} -- with no incident vertices. 

    Define $\gk$ to be the set of all $\fc$-gadgets, and $\gk(m,d) \subset \gk$ to be the set 
    of gadgets with $m$ left and $d$ right dangling edges. The $m$ left dangling edges
    $\ell_0,\ldots,\ell_{m-1}$ and $d$ right dangling edges $r_{0}, \ldots, r_{d-1}$ 
    are drawn on the left and right of $\vk$, respectively, with dangling ends (the ends
    with no endpoint) from top to bottom, so in overall counterclockwise cyclic order
    $\ell_0, \ldots, \ell_{m-1}, r_{d-1}, \ldots, r_0$ around $\vk$.
    The two dangling ends of a wire are part of the same dangling edge, so
    $\ell_0, \ldots, \ell_{m-1}, r_{d-1}, \ldots, r_0$ may not be distinct.
\end{definition}
See \autoref{fig:gadgetops} for examples of gadgets.
Connecting the two ends of a wire, we obtain a vertexless loop.
A gadget in $\gk(m,d)$ defines a $(m+d)$-ary signature in flattened form, with dangling edges representing
inputs, as follows (cf. \eqref{eq:holant}):
\begin{definition}[$M(\vk)$]
    \label{def:gadget_sig_matrix}
    Define the \emph{signature matrix} $M(\vk) \in \rr^{q^m \times q^d}$ of $\vk \in \gk(m,d)$ by
    \[
        M(\vk)_{\vx,\vy} = \sum_{\substack{\sigma: E(\vk) \to [q]
        \\ \forall i: \sigma(\ell_i) = x_i 
        \\ \forall j: \sigma(r_j) = y_{j}}}
        \prod_{v \in V} F_v(\sigma|_{\delta(v)})
        ~~\text{ for $\vx \in [q]^m$ and $\vy \in [q]^d$}.
    \]
\end{definition}
In other words, $M(\vk)_{\vx,\vy}$ equals the Holant value of $\vk$ when the left and right dangling
edges are fixed to the domain elements in $\vx$ and $\vy$, respectively.
As a gadget with one left and one right dangling end, a wire
has siganture matrix $I$ (the identity), as its left and right inputs must agree.
If $\vk \in \gk(m,d)$ is a gadget consisting
of a single vertex assigned $F$ with $m$ and $d$ incident left and right dangling edges, 
respectively, ordered to match the input order of $F$, then $M(\vk) = F^{m,d}$.
In general, for $\vk \in \gk(m,d)$, there is a unique $F \in (\rr^q)^{\otimes m+d}$, called the \emph{signature of} $\vk$,
such that $M(\vk) = F^{m,d}$. Note that $F$ does not depend on the
particular left/right partition (i.e. choice of $m$ and $d$) of a fixed cyclic order 
of $\vk$'s dangling edges.
\begin{definition}[$\circ, \otimes, \top$]
    Define the following three operations on gadgets:
    \begin{itemize}
        \item For $\vk \in \gk(m,d)$ and $\vl \in \gk(d,r)$, construct $\vk \circ \vl \in \gk(m,r)$
            by connecting the dangling ends of $r_i \in R(\vk)$ and $\ell_i \in L(\vl)$ for $i \in [d]$.
        \item For $\vk \in \gk(m_1,d_1)$ and $\vl \in \gk(m_2,d_2)$, construct 
            $\vk \otimes \vl \in \gk(m_1+m_2,d_1+d_2)$ as the disjoint union of $\vk$ and $\vl$,
            placing $\vk$ above $\vl$. Interleave $\vk$ and $\vl$'s dangling edges into
            an overall cyclic order: $\vk$ left, then $\vl$ left, then $\vl$ right, then $\vk$ right.
        \item For $\vk \in \gk(m,d)$, construct $\vk^{\top} \in \gk(d,m)$ 
            by exchanging the roles of left and right dangling edges and reversing
            the overall cyclic dangling edge order (visually, horizontally reflect $\vk$).
    \end{itemize}
\end{definition}
See \autoref{fig:gadgetops}. As one might expect,
the gadget operations $\circ, \otimes, \top$ induce the respective operations -- composition, 
Kronecker product, transpose -- on their signature matrices. See e.g. 
\cite[Section 1.3]{cai_chen_2017}.
\begin{figure}[ht!]
    \begin{center}
    \begin{subfigure}{0.8\textwidth}
        \centering
\includegraphics[width=0.6\textwidth]{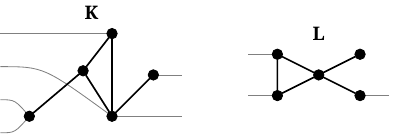}
    \end{subfigure}

    \begin{subfigure}{0.4\textwidth}
\includegraphics[width=0.9\textwidth]{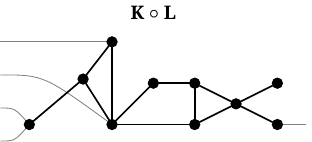}
        \vspace{0.5cm}
    \end{subfigure}
    \begin{subfigure}{0.25\textwidth}
\includegraphics[width=0.9\textwidth]{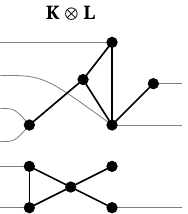}
    \end{subfigure}
    \hspace{0.6cm}
    \begin{subfigure}{0.25\textwidth}
\includegraphics[width=0.9\textwidth]{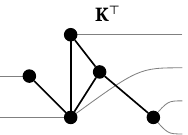}
        \vspace{1cm}
    \end{subfigure}
    \end{center}
    \caption{Operations on gadgets $\vk \in \mathfrak{G}(4,2)$ and $\mathbf{L} \in \mathfrak{G}(2,1)$. 
        Dangling edges are drawn thinner than internal edges.}
    \label{fig:gadgetops}
\end{figure}

\begin{definition}[$\langle\cdot,\cdot\rangle,\|\cdot\|$]
    For real-valued $\fc$ and $n$-ary gadgets $\vk$ and $\vl$, construct the signature grid
    $\langle \vk,\vl \rangle$ by connecting the $i$th dangling edges
    of $\vk$ and $\vl$, for $i \in [n]$.
    If $\vk$ and $\vl$ have signatures $K$ and $L$, then define $\langle K,L \rangle
    := \holant_{\langle \vk,\vl\rangle} = \langle K^{n,0}, L^{n,0}\rangle$ (the standard
    inner product on $\rr^{q^n}$). Define $\|F\| := \sqrt{\langle F,F\rangle} = \sqrt{\sum_{\vx} F_{\vx}^2}$.
\end{definition}

\subsection{The Holant Theorem}
\label{sec:holanttheorem}
Throughout this work, we will be considering pairs of multisets of signatures, called simply
\emph{signature sets}, usually denoted $\fc$ and $\gc$. 
Such pairs are assumed to be \emph{similar}, meaning they have the same domain size $q$ and
they come equipped with a multiset bijection between $\fc$ and $\gc$ such that, 
for $n$-ary $F \in \fc$, the image $G \in \gc$ 
of $F$, called the signature \emph{corresponding} to $F$, has the same arity $n$.
This correspondence bijection, usually given implicitly along with a similar pair $\fc,\gc$, 
is denoted by $\leftrightsquigarrow$. That is, $\fc \ni F \leftrightsquigarrow G \in \gc$ 
indicates that $G$ corresponds to $F$. We consider multisets, rather than sets, of signatures
because we will wish to preserve the bijective correspondence between $\fc$ and $\gc$ when 
adding signatures to $\fc$ and $\gc$ that may already be present.
\begin{definition}[$\vk_{\fc\to\gc}, \Omega_{\fc\to\gc}$]
    For similar signature sets $\fc$ and $\gc$ with correspondence bijection 
    $\leftrightsquigarrow$ and gadget $\vk \in \gk$, define the gadget
    $\vk_{\fc\to\gc} \in \mathfrak{G}_{\gc}$ by replacing every $F \in \fc$ assigned to a 
    vertex in
    $\vk$ with the $G \in \gc$ such that $F \leftrightsquigarrow G$. 
    If $\vk$ has zero dangling edges then it is an
    $\fc$-grid $\Omega$, and is transformed to a $\gc$-grid $\Omega_{\fc\to\gc}$.
\end{definition}
We next introduce the concept of signature transformation.
\begin{definition}[$HF,H\fc$]
    Define an action of the group $\text{GL}_q(\rr)$ of invertible $q \times q$ matrices on the set 
    $T(\rr^q)$ of domain-$[q]$ signatures as follows: for $H \in \text{GL}_q(\rr)$ and 
    $F \in (\rr^q)^{\otimes n}$ with signature vector $f$, let $HF\in (\rr^q)^{\otimes n}$ be the signature whose signature vector is
    $H^{\otimes n}f$ -- that is, $(HF)^{n,0} = H^{\otimes n}f$ (see \autoref{fig:pivot_h}).

    For a signature set $\fc \subset T(\rr^q)$, define $H\fc := \{HF \mid F \in \fc\}$.
\end{definition}
\begin{figure}[ht!]
    \center
\includegraphics[width=0.8\textwidth]{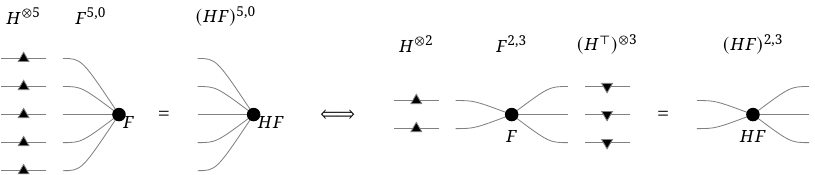}
    \caption{The gadget perspective on $H^{\otimes n}f = (HF)^{n,0}$, or equivalently 
    $H^{\otimes m} F^{m,d} (H^\top)^{\otimes d} = (HF)^{m,d}$. Each triangle is a vertex
    assigned $H$ as a binary signature.}
    \label{fig:pivot_h}
\end{figure}

Holant problems were originally motivated by the following theorem, a powerful reduction tool
proved by Valiant \cite{valiant}.
For $F \in (\rr^q)^{\otimes n}$ and $A \in \text{GL}_q$, define $FA$ similarly
to $AF$, by $(FA)^{0,n} = F^{0,n}A^{\otimes n}$.
\begin{theorem}[The Holant Theorem]
    For any $(\fc \mid \fc')$-grid $\Omega$ and matrix $A \in \text{GL}_q$,
    \[
        \holant_{\Omega}(\fc \mid \fc') = \holant_{\Omega'}
        (A\fc \mid \fc' A^{-1}),
    \]
    where $\Omega' = \Omega_{(\fc|\fc') \to (A\fc | \fc' A^{-1})}$.
    \label{thm:holant}
\end{theorem}
Xia \cite{xia} conjectured that the converse of \autoref{thm:holant} holds as long as
$\fc$ or $\fc'$ contains a signature with arity greater than one -- that is, if 
$\holant_{\Omega}(\fc \mid \fc') = \holant_{\Omega_{(\fc|\fc') \to (\gc | \gc')}}(\gc \mid \gc')$
for every $(\fc \mid \fc')$-grid $\Omega$, then there is an $A \in \text{GL}_q$ 
such that $\gc = A\fc$ and $\gc' = \fc' A^{-1}$. However, Cai, Guo, and Williams
\cite[Section 4.3]{cai_complete_2016} observe that this conjecture is false. They consider the problems
\[
    \holant([0,1,0] \mid [a,b,1,0,0]) ~~\text{and}~~ \holant([0,1,0] \mid [0,0,1,0,0])
\]
for any $a,b$ not both 0,
where $[0,1,0]$ and $[a,b,1,0,0]$ are symmetric signatures on the Boolean domain $\{0,1\}$ 
of arity $n=2,4$, respectively, specified by their values on input strings of Hamming
weight $0$ through $n$.
For both problems, $[0,1,0]$ is only nonzero on edge assignment $\sigma$ if $\sigma$ assigns its two
incident edges opposite values. Thus, unless its contribution to the total Holant value is 0, 
$\sigma$ must assign 0 to exactly half the edges in $\Omega$, and 1 to the other half. Then, since
every $[a,b,1,0,0]$ evaluates to zero
whenever it receives more 1 inputs than 0 inputs, no nonzero assignment $\sigma$ assigns any
$[a,b,1,0,0]$ more 0 inputs than 1 inputs. Therefore, in this context,
$[a,b,1,0,0]$ is indistinguishable from $[0,0,1,0,0]$, so the hypothesis of Xia's conjecture is satisfied. 
However, there is no $A \in \text{GL}_2$ satisfying $A[0,1,0] = [0,1,0]$ and 
$[a,b,1,0,0]A^{-1} = [0,0,1,0,0]$. To see this, if $A = \begin{bmatrix} a & b \\ c & d\end{bmatrix}$, then, with the signature vector of $[0,1,0]$ being $(0,1,1,0)^\top$,
\[
    A^{\otimes 2}
    \begin{bmatrix} 0 \\ 1 \\ 1 \\ 0 \end{bmatrix} =
    \begin{bmatrix} 
        a^2 & ab & ba & b^2 \\
        ac & ad & bc & bd \\
        ca & cb & da & db \\
        c^2 & cd & dc & d^2
    \end{bmatrix}
    \begin{bmatrix} 0 \\ 1 \\ 1 \\ 0 \end{bmatrix}
    = \begin{bmatrix} 0 \\ 1 \\ 1 \\ 0 \end{bmatrix}
    \iff ab = cd = 0 \text{ and } ad + bc = 1
    \implies \begin{cases} a = d = 0 & \text{or} \\ b = c = 0\end{cases}.
\]
If $b=c=0$ then $A$ simply rescales a signature's entries, and if $a=d=0$ then
$A$ rescales a signature's entries and exchanges the roles of $0$ and $1$, which has the effect of
reversing the entries in the $[\cdot,\ldots,\cdot]$ notation. Thus $[0,0,1,0,0]A = [0,0,*,0,0]
\neq [a,b,1,0,0]$.

This counterexample exists due to 
the bipartiteness enforced by $\holant(\fc \mid \gc)$, because
\[
    \left\langle [0,0,1,0,0],[0,0,1,0,0]\right\rangle = \|[0,0,1,0,0]\|^2
    \neq \|[a,b,1,0,0]\|^2 = \left\langle [a,b,1,0,0],[a,b,1,0,0] \right\rangle,
\]
so $[a,b,1,0,0]$ and $[0,0,1,0,0]$ are not indistinguishable on
general (non-bipartite) signature grids. To avoid such counterexamples, we consider the
following well-known form of \autoref{thm:holant} that applies to non-bipartite grids.
An edge in an $\fc$-grid $\Omega$, viewed on its own, is a wire,
with signature $I$. Therefore we can, without changing the Holant value, replace each edge in 
$\Omega$ by a binary gadget with a single vertex assigned $I$, splitting the edge into a path
of length two. Then $\Omega$ is a signature grid in the context of
$\holant(\fc \mid I)$, so, by \autoref{thm:holant}, for any invertible $H$,
$\holant_\Omega(\fc \mid I) = \holant_{\Omega'}(H\fc \mid \{I\}H^{-1})$. 
Let $O(q)$ be the group of $q \times q$ orthogonal matrices.
If $H \in O(q)$ then, as in \autoref{fig:pivot_h}, $(I^{2,0} (H^{-1})^{\otimes 2})^{1,1}
= (H^{-1})^\top I H^{-1} = H H^{-1} = I$. Thus $\{I\} H^{-1} = \{I\}$ 
(see also \autoref{thm:fft} below), so
$\holant_\Omega(\fc \mid I) = \holant_{\Omega'}(H\fc \mid I)$. Then $\holant(H\fc \mid I)$
is again equivalent to $\holant(H\fc)$, so we have proved the following.
\begin{corollary}[The Orthogonal Holant Theorem]
    \label{cor:holanttheorem}
    If $\gc = H\fc$ for orthogonal $H$, then $\holant_\Omega(\fc) = \holant_{\Omega_{\fc\to\gc}}(\gc)$ for every $\fc$-grid $\Omega$.
\end{corollary}
Xia \cite{xia} also considers the converse of \autoref{cor:holanttheorem}, and proves that it holds
for specific $\fc$ and $\gc$ consisting of real-valued symmetric signatures with small domain
and/or arity.
The main result of this work is the converse of \autoref{cor:holanttheorem} for any
sets $\fc$ and $\gc$ of real-valued signatures, with no restrictions.
\begin{theorem}[Main Result]
    \label{thm:result}
    The following are equivalent for similar sets $\fc$, $\gc$ of real-valued signatures. 
    \begin{enumerate}[label=(\roman*)]
        \item $\holant_{\Omega}(\fc) = \holant_{\Omega_{\fc\to\gc}}(\gc)$
        for every $\fc$-grid $\Omega$.
        \item There is a real orthogonal matrix $H$ such that $H\fc = \gc$.
    \end{enumerate}
\end{theorem}
Call $\fc$ and $\gc$ satisfying (i) \emph{Holant-indistinguishable}, and call $\fc$ and
$\gc$ satisfying (ii) \emph{ortho-equivalent}. 
The following two properties follow directly from the definitions and \autoref{cor:holanttheorem},
but, as they will
prove useful throughout the proof of \autoref{thm:result}, we state them explicitly.
\begin{proposition}
    For any orthogonal $H_1$ and $H_2$, $\fc$ and $\gc$ are 
    ortho-equivalent/Holant-indistinguishable if and only if
    $H_1\fc$ and $H_2\gc$ are ortho-equivalent/Holant-indistinguishable, respectively.
    \label{prop:transform}
\end{proposition}
\begin{proposition}
    For any similar pair $\fc$,$\gc$ and additional similar pair $\fc'$,$\gc'$ 
    of signature sets, the similarity bijections between $\fc,\gc$ and between $\fc',\gc'$
    induce a similarity bijection between the disjoint unions 
    $\fc \sqcup \fc'$ and $\gc \sqcup \gc'$, and if
    $\fc \sqcup \fc'$ and $\gc \sqcup \gc'$ 
    are ortho-equivalent, then $\fc$ and $\gc$ are ortho-equivalent.
    \label{prop:union}
\end{proposition}   

\begin{remark}\label{rem:complex}
Although we will focus on real-valued signatures, the general and orthogonal 
Holant Theorems hold for complex-valued signatures. However,
\autoref{thm:result} does not hold for general sets $\fc$ and $\gc$ of complex-valued
signatures, even when we allow $H$ to be complex. For example, Draisma and Regts
\cite{draisma_tensor_2013} consider the \emph{vanishing} unary signature $F \in \mathbb{C}^{[2]^1}$
defined by $F_0 = 1$ and $F_1 = i$. Any $F$-grid $\Omega$ with at least
one vertex satisfies $\holant_\Omega(F) = 0$, as $\Omega$ is a disjoint union of $K_2$ complete graphs,
with each component having value $f^\top f = [1,i]^\top [1,i] = 0$. 
Thus $F$ is Holant-indistinguishable from 0, but
there is no orthogonal matrix $H$, real or complex, satisfying $Hf = H [1,i]^\top = [0,0]^\top$.
However, as we will discuss briefly in \autoref{sec:complex}, \autoref{thm:result} holds for
complex $\fc$ and $\gc$ under certain additional assumptions.
\end{remark}

\subsection{Block Matrices and Signatures}
\label{sec:block}
\begin{definition}
    \label{def:block}
    Let $\ic$ be an index/domain set, and $X \sqcup Y = \ic$ be a nontrivial partition of $\ic$.
    \begin{enumerate}
    \item
    For a matrix $H \in \rr^{\ic \times \ic}$ and $R,C \in \{X,Y\}$, let $H|_{R,C} \in \rr^{R\times C}$ be the
    submatrix of $H$ with rows indexed by $R$ and columns indexed by $C$. Up to row and
    column reordering, $H$ is the block matrix
    \[
        H = \begin{bmatrix} H|_{X,X} & H|_{X,Y} \\ H|_{Y,X} & H|_{Y,Y}\end{bmatrix}.
    \]
    \item
    More generally, for a signature/tensor $F \in \rr^{\ic^n}$ and $\vr \in \{X,Y\}^n$, define
    $F|_{\vr} \in \rr^{\vr}$ (where we identify $\vr$ with the set $\prod_{i=1}^n R_i$) to be the subtensor of $F$ with $i$th
    input restricted to $R_i$.
    For signature set $\fc$, let $\fc|_{\vr} := \{F|_{\vr}: F \in \fc\}$.

    Abbreviate $F|_X := F|_{X^n}$ and $\fc|_X := \fc|_{X^n}$.
    \item 
    Let $F^{m,d} \in \rr^{\ic^m \times \ic^d}$ be a signature matrix, and let 
    $\vr \in \{X,Y\}^m$ and $\vc \in \{X,Y\}^d$. Define 
    $F^{m,d}|_{\vr,\vc} \in \rr^{\vr \times \vc}$
    as the submatrix of $F^{m,d}$ with rows indexed by $\vr$ and columns indexed by
    $\vc$ (in other words, $F^{m,d}|_{\vr,\vc} = (F|_{\vr\vc})^{m,d}$, where
    $\vr\vc$ is the concatenation of $\vr$ and $\vc$).
    \end{enumerate}
\end{definition}

\begin{definition}[$\oplus$]
\label{def:oplus-f-g}
    Let $F,G$ be $n$-ary signatures on domains $V(F)$, $V(G)$, both of size $q$.
    The \emph{direct sum} $F \oplus G$ of $F$ and $G$ is an
    $n$-ary signature on domain $V(F) \sqcup V(G)$ defined by
    \[
        (F \oplus G)_{\vx} = \begin{cases} 
            F_{\vx} & \vx \in V(F)^n \\
            G_{\vx} & \vx \in V(G)^n \\
            0 & \text{otherwise}
        \end{cases}
        \qquad\text{for } \mathbf{x} \in (V(F) \sqcup V(G))^n.
    \]
    For signature sets $\fc$ and $\gc$, define 
    $\fc \oplus \gc = \{F \oplus G \mid \fc \ni F \leftrightsquigarrow G \in \gc\}$.
\end{definition}
We will frequently apply \autoref{def:block} to the partition $V(\fc) \sqcup V(\gc)$ of the domain
of $\fc\oplus\gc$.
It follows directly from the definitions that, for $\fc \ni F \leftrightsquigarrow G \in \gc$ and any 
$m+d = \arity(F)$,
\begin{equation}
    (F \oplus G)^{m,d}|_{\vr,\vc} = 
    \begin{cases}
        F^{m,d} & \vr = V(\fc)^m \wedge \vc = V(\fc)^d \\ G^{m,d} & \vr = V(\gc)^m \wedge \vc = V(\gc)^d \\ 0 & \text{otherwise}.
    \end{cases}
    \label{eq:oplus_index}
\end{equation}
\autoref{sec:appendix_block} states and proves some additional necessary results involving 
actions of block matrices on block signatures.
These additional results confirm that these block structures interact as
one would expect, analogously to ordinary block matrix algebra.


\section{Domain Induction: The Proof of \autoref{thm:result}}
\label{sec:proof}
Throughout the proof of \autoref{thm:result}, we would like to assume that $\fc$ and $\gc$
contain not only the signatures of all $\fc$ and $\gc$-gadgets, but all linear combinations
of such signatures. There are two principal reasons for this. First, the
invariant-theoretic machinery in \autoref{sec:duality} demands that our signature sets be
vector spaces closed under $\circ,\otimes,\top$. 
Second, for the domain induction argument we will need $\fc$ and $\gc$ to contain the
matrices (binary signatures) projecting onto certain subdomains. We obtain these projectors
via polynomial interpolation, which requires both matrix composition
(gadget composition) and linear combinations.
These considerations motivate the following definition.
\begin{definition}[$\qk(m,d), {[\fc]}$]
    \label{def:quantum_gadget}
    A \emph{$(m,d)$-quantum $\fc$-gadget} is a formal (finite) $\rr$-linear combination of gadgets in
    $\gk(m,d)$. Let $\qk(m,d)$ be the set of all $(m,d)$-quantum $\fc$-gadgets,
    and $\qk = \bigcup_{m,d}\qk(m,d)$.

    Extend the signature matrix function $M$ linearly to $\qk$. 
    Then define the \emph{quantum gadget closure} $[\fc]$ of $\fc$ as the
    multiset of signatures obtainable from $\fc$ by quantum gadget construction:
    \[
        [\fc] = \bigsqcup_{\vk \in \qk} (\text{signature of $\vk$}).
    \]
\end{definition}
In particular, $[\fc]$ is in bijection with $\qk$, so,
if $\fc$ and $\gc$ are similar, then $[\fc]$ and
$[\gc]$ are similar, with the signature of $\vk \in \qk$ corresponding to the signature of
$\vk_{\fc\to\gc} \in \mathfrak{Q}_{\gc}$.
A signature set $\fc \subset T(\rr^q)$ is \emph{quantum-gadget-closed} if $\fc = [\fc]$ as
sets.

The proof of \autoref{thm:result} now proceeds roughly as follows. First, \autoref{prop:closure}
lets us replace $\fc$ and $\gc$ by their quantum gadget closures. Then we use the
invariant-theoretic duality theory from \autoref{sec:duality}, in the form of
\autoref{lem:d_intertwine}, to obtain a nonzero diagonal matrix $D$ \emph{intertwining}
(roughly, non-invertibly transforming between) $\fc$ and $\gc$. If $D$ is a scalar multiple
of the identity matrix, then this multiple must be $\pm 1$ and $D$ is the desired
orthogonal transformation. Otherwise, from $D$ we
interpolate the projectors onto disjoint eigenspaces of $D$ -- two subdomains forming a 
nontrivial partition of the full domain $[q]$. 
These projectors pass indistinguishability to the two subdomains (\autoref{lem:restriction}).
Then, by induction on the domain size, we obtain orthogonal transformations between $\fc$
and $\gc$ on the two subdomains. Finally, the most involved step is combining these
subdomain transformations into
an orthogonal transformation on the whole domain in \autoref{lem:induction}.

The following proposition, a
nonplanar, orthogonal version of \cite[Lemmas 31 and 32]{cai_planar_2023}, lets us assume
$\fc$ and $\gc$ are quantum-gadget-closed when proving \autoref{thm:result}.
We defer the proof until after we develop appropriate machinery in \autoref{sec:duality} but,
intuitively, the first point holds because corresponding $[\fc]$ and $[\gc]$ grids decompose
into corresponding quantum $\fc$ and $\gc$ grids, and the second point holds because
the action of an orthogonal $H$ respects the operations used to construct gadgets.
\begin{restatable}{lemma}{closurelemma}
  \label{prop:closure}\RestateRemark
  For any signature sets $\fc$ and $\gc$,
    \begin{enumerate}
        \item $\fc$ and $\gc$ are Holant-indistinguishable iff $[\fc]$ and $[\gc]$ are
            Holant-indistinguishable.
        \item For any orthogonal $H$, $H\fc = \gc$ iff $H[\fc] = [\gc]$
    (in particular, $H [\fc] = [H\fc]$).
    \end{enumerate}
\end{restatable}


In the context of \#CSP, $\eq$ enables the gadget-construction of entrywise 
products of arbitrary signatures. This is a basis for Vandermonde interpolation, a powerful 
technique in the study of counting problems.
In the general Holant setting, we cannot construct the entrywise product of arbitrary signatures,
but, for binary signatures whose matrix forms are diagonal, composition is equivalent to entrywise
product. We use the following form of Vandermonde interpolation, which follows directly from the similar
\cite[Lemma 2.3]{grohe_homomorphism_2022}.
\begin{proposition}
    \label{prop:vandermonde}
    Let $V \subset \rr^q$ be a vector space closed under entrywise product and containing the
    all-ones vector. For $v \in V$ and $a \in \rr$, define the indicator vector $v^a \in \{0,1\}^q$ by
    $(v^a)_x = 1$ if $v_x = a$ and $(v^a)_x = 0$ otherwise. Then $v^a \in V$ for every
    $v \in V$ and $a \in \rr$.
\end{proposition}
\autoref{prop:vandermonde} will give us the following matrices, whose diagonals we identify
with $\rr^q$:
For $Z \subset [q]$, define the diagonal matrix (binary signature)
$\mathds{1}_Z \in \{0,1\}^{q \times q}$ by $(\mathds{1}_Z)_{x,x} = 1$ if $x \in Z$, and $(\mathds{1}_Z)_{x,x} = 0$ otherwise.
\begin{lemma}
    \label{lem:restriction}
    Let $\fc$ and $\gc$, on domain $[q]$, be Holant-indistinguishable, and let $Z \subset [q]$. If $\fc$ and
    $\gc$ contain corresponding copies of $\mathds{1}_Z$, then $\fc|_Z$ and $\gc|_Z$ are 
    Holant-indistinguishable.
\end{lemma}
\begin{proof}
    Let $\Omega$ be a $\fc|_Z$-grid.
    Construct a $\fc$-grid $\Omega'$ from $\Omega$ by replacing every signature
    $F|_Z \in \fc|_Z$ with the corresponding $F \in \fc$, then replacing every edge with a degree-2 vertex
    assigned 
    $\mathds{1}_Z \in \fc$. Since $\mathds{1}_Z$ is 0 when given any input not in $Z$,
    any edge assignment sending an edge outside of $Z$ contributes 0 to
    $\holant_{\Omega'}(\fc)$. Furthermore, on inputs from $Z$, $\mathds{1}_Z$ acts identically to an 
    edge (as $I$) in a $\holant(\fc|_Z)$ grid.
    Similar reasoning applies to $\gc$ and $\gc|_Z$, so
    \[
        \holant_{\Omega}(\fc|_Z) = \holant_{\Omega'}(\fc) = \holant_{\Omega'_{\fc\to\gc}}(\gc) 
        =\holant_{\Omega_{\fc|_Z\to\gc|_Z}}(\gc|_Z). \qedhere
    \]
\end{proof}

The following definition and its applications below borrow a simple but powerful idea of Shao and
Cai \cite[Section 8.2]{shao}: isolating all vertices of an
$\fc \cup \{F\}$-grid $\Omega$ assigned $F$, the rest of $\Omega$ is an $\fc$-gadget. 
\begin{definition}[Subgadget, $\overline{\vk}$]
    Let $\vj$ be a gadget. A \emph{subgadget} $\vk \subset \vj$ induced by a subset $U \subset V(\vj)$ 
    of vertices of $\vj$ is a gadget composed of the vertices in $U$ and 
    \emph{all} of their incident edges: internal edges of $\vj$ incident to 
    exactly one vertex in $U$ become new dangling edges of $\vk$.
    For any $\vk \subset \vj$, there is a
    unique (up to permutation of dangling edges) subgadget
    $\overline{\vk} \subset\vj$, induced by $V(\vj) \setminus U$, called the 
    \emph{complement} of $\vk$, such that, upon reconnecting the new dangling edges of $\vk$ and
    $\overline{\vk}$, we recover $\vj$.

    We often take $\vj$ to be a signature grid (0-ary gadget) $\Omega$, in which case
    $\Omega = \left\langle \vk,\overline{\vk} \right\rangle$.
\end{definition}

We now prove the key inductive lemma. Given a diagonal matrix $D$ capable of separating the
domain $[q]$, we apply \autoref{lem:restriction} to obtain indistinguishability on both
subdomains, apply induction to obtain orthogonal transformations on both subdomains, then
show how to stitch these subdomain transformations together into a transformation on the
whole domain.
\begin{lemma}
    \label{lem:induction}
    Let $\fc$ and $\gc$ be signature sets on domain $[q]$, and suppose \autoref{thm:result} holds for
    all $\fc'$, $\gc'$ on domain smaller than $q$. If $\fc$ and $\gc$ are Holant-indistinguishable
    and contain corresponding copies of a diagonal matrix (binary signature)
    $D \not\in \spn(I)$, then $\fc$ and $\gc$ are ortho-equivalent. 
\end{lemma}
\begin{proof}
    By \autoref{prop:closure}, we may replace $\fc$ and $\gc$ with $[\fc]$ and $[\gc]$ to assume $\fc$
    and $\gc$ are quantum-gadget-closed.
    Since $D \not\in \spn(I)$, there exist $x,y \in [q]$ such that $D_{x,x} \neq D_{y,y}$,
    so the sets 
    \[
        X = \{z \in [q]: D_{z,z} = D_{x,x}\} \text{ and } Y = [q] \setminus X
    \]
    are a nontrivial partition of $[q]$. Since $\fc$ is quantum-gadget-closed, it contains 
    the signature $I$ of the wire. 
    Consider the subalgebra $\langle D,I\rangle_{+,\circ} \subset [\fc] = \fc$.
    Since $D$ and $I$ are diagonal, composition $\circ$ is equivalent to entrywise multiplication
    in $\langle D,I\rangle_{+,\circ}$. Therefore, by \autoref{prop:vandermonde} (identifying the
    matrix diagonals with $\rr^q$),
    we have $\mathds{1}_X, \mathds{1}_Y \in\langle D,I\rangle_{+,\circ} \subset \fc$.
    Applying the same interpolation in $\gc$, we obtain corresponding copies of 
    $\mathds{1}_X, \mathds{1}_Y \in \gc$. 
    Now, applying \autoref{lem:restriction} with $Z := X$, we conclude
    $\fc|_X$ and $\gc|_X$ are Holant-indistinguishable. Furthermore, $|X| < q$, so, by assumption,
    there is an orthogonal matrix $H_X \in \rr^{X \times X}$ satisfying 
    \begin{equation}
        H_X \fc|_X = \gc|_X.
        \label{eq:hx}
    \end{equation}
    For $b \in [q]$, let $\Delta_b \in \rr^q$ be the unary \emph{pinning signature} defined by
   \[\Delta_b(x) = 
      \begin{cases} 1 & b = x \\ 0 & b \neq x \end{cases} \text{ for } x \in [q].
     \]
    Call signatures in $\{\Delta_b \mid b \in X\}$ \emph{$X$-pins}.
    Let $I_Y$ be the identity operator on $\rr^{Y\times Y}$ and define
    \begin{equation}
        \fc' = \fc \sqcup \{\Delta_b \mid b \in X\} \text{ and } \gc' = \big((H_X^{-1} \oplus I_Y)\gc\big) \sqcup
        \{\Delta_b \mid b \in X\}.
        \label{eq:fcprime}
    \end{equation}
    \begin{claim}
        $\fc'$ and $\gc'$ are Holant-indistinguishable.
        \label{claim:equiv}
      \end{claim}
      \begin{subproof}
            To see this, let $\Omega$ be a connected $\fc'$-grid. If $\Omega$ contains no 
    $X$-pins then it is an $\fc$-grid, and 
    $\Omega_{\fc'\to\gc'}$ is the corresponding $(H_X^{-1} \oplus I_Y)\gc$-grid, so by assumption and \autoref{cor:holanttheorem}, $\holant_\Omega = \holant_{\Omega_{\fc'\to\gc'}}$.
    If $\Omega$ contains two adjacent $X$-pins, then, since $\Omega$ is connected, its underlying 
    multigraph
    consists only of these two vertices. $X$-pins in $\fc'$ correspond to identical $X$-pins in $\gc'$, so
    again $\holant_\Omega = \holant_{\Omega_{\fc'\to\gc'}}$. Otherwise, $\Omega$ contains $p$ 
    pairwise non-adjacent vertices assigned $X$-pins.
    Let $\vk \subset \Omega$ be a subgadget induced by these $p$ vertices, so
    $\vk$'s signature is $\bigotimes_{i=1}^p \Delta_{b_i}$ for some $b_1,\ldots,b_p \in X$.
    Since $\vk$ includes \emph{all}
    the vertices in $\Omega$ assigned $X$-pins, its complement $\overline{\vk}$ is an $\fc$-gadget.
    See \autoref{fig:pin_gadget}. 
    \begin{figure}[ht!]
        \centering
\includegraphics[width=0.7\textwidth]{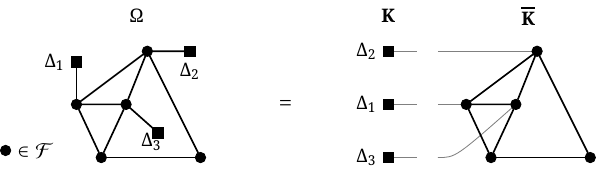}
        \caption{An $\fc'$-grid $\Omega$ and its pin-induced subgadget $\vk$ and complement $\overline{\vk}$}
        \label{fig:pin_gadget}
    \end{figure}
    Hence the signature $F$ of $\overline{\vk}$ is in $\fc$, as $\fc$ is quantum-gadget-closed.
    From $\Omega = \langle \vk,\overline{\vk} \rangle$, we obtain
    \[
        \holant_\Omega = \left\langle \bigotimes_{i=1}^p \Delta_{b_i}, F \right\rangle = F_{\vb},
    \]
    where $\vb = b_1b_2 \ldots b_p$.
    A similar calculation, using the fact that $\gc$ is also quantum-gadget-closed, 
    shows $\holant_{\Omega_{\fc'\to\gc'}} = ((H^{-1}_X \oplus I_Y)G)_{\vb}$ for
    $\gc \ni G \leftrightsquigarrow F$. Then, since $\vb \in X^p$, we have, using \eqref{eq:hx},
    \[
        \holant_\Omega = F_{\vb} = (F|_X)_{\vb} = (H_X^{-1} G|_X)_{\vb} = 
        (((H^{-1}_X \oplus I_Y)G)|_X)_{\vb}
        = ((H^{-1}_X \oplus I_Y)G)_{\vb}
        = \holant_{\Omega_{\fc'\to\gc'}}
    \]
    (where the fourth equality uses \autoref{cor:on_diag_block} with $K := G$ and
    $H := H_X^{-1} \oplus I_Y$).
    So $\fc'$ and $\gc'$ are Holant-indistinguishable, completing the proof of \autoref{claim:equiv}.
     \end{subproof}

    By \autoref{prop:closure}, $[\fc']$ and $[\gc']$ are also Holant-indistinguishable
    and, by \autoref{prop:transform} and \autoref{prop:union}, 
    it suffices to show $[\fc']$ and $[\gc']$ are
    ortho-equivalent to complete the proof. 
    Note that $[\fc']$ 
    and $[\gc']$ still contain $\mathds{1}_Y$ (as $\mathds{1}_Y$, being 0 on $X$, is unaffected
    by the transform $H_X^{-1} \oplus I_Y$), so we may
    apply \autoref{lem:restriction} to conclude
    $\big([\fc']\big)\big|_Y$ and $\big([\gc']\big)\big|_Y$ are Holant-indistinguishable. 
    Again, since $|Y| < q$, there is an orthogonal
    $H_Y \in \rr^{Y \times Y}$ such that $H_Y\big([\fc']\big)\big|_Y = \big([\gc']\big)\big|_Y$. Define
    \begin{equation}
        \fc'' = [\fc'] \sqcup \{\Delta_b \mid b \in Y\} \text{ and }
        \gc'' = \big((I_X \oplus H_Y^{-1}) [\gc'] \big) \sqcup \{\Delta_b \mid b \in Y\}.
    \label{eq:fcprimeprime}
    \end{equation}
    By \autoref{prop:transform} and \autoref{prop:union}, it suffices to show that $\fc''$ and $\gc''$
    are ortho-equivalent. 
    As $[\fc']$ and $[\gc']$ are, like $\fc$ and $\gc$, quantum-gadget-closed and 
    Holant-indistinguishable, we repeat the proof of \autoref{claim:equiv},
    with \eqref{eq:fcprimeprime} in place of \eqref{eq:fcprime}, to show that
    $\fc''$ and $\gc''$ are Holant-indistinguishable.
    Observe that, by definition, $\fc''$ and $\gc''$ contain corresponding
    copies of all $Y$-pins $\{\Delta_b \mid b \in Y\}$.
    Furthermore, $\fc'' \supset [\fc'] \supset \fc' \supset \{\Delta_b \mid b \in X\}$
    and 
    \[
        \gc'' \supset (I_X \oplus H_Y^{-1})[\gc'] \supset (I_X \oplus H_Y^{-1})\gc' \supset (I_X \oplus H_Y^{-1})\{\Delta_b \mid b \in X\} = \{\Delta_b \mid b \in X\},
    \]
    where the final equality holds because the $X$-pins are zero on $Y$, so are unaffected by the
    transform $(I_X \oplus H_Y^{-1})$. Thus $\fc''$ and $\gc''$ contain corresponding copies of all
    pins $\Delta_b$ for $b \in [q]$. 
    We claim that this implies that $\fc'' = \gc''$.
    To see this, consider any $\fc'' \ni F \leftrightsquigarrow G \in \gc''$
    of common arity $n$. For any $\vx \in [q]^n$, since $\fc''$ and $\gc''$ are Holant-indistinguishable, 
    we have
    \[
        F_{\vx} = \left\langle F,\bigotimes_{i=1}^n \Delta_{x_i} \right\rangle 
        = \left\langle G,\bigotimes_{i=1}^n \Delta_{x_i} \right\rangle = G_{\vx}.
    \]
    Thus $F = G$ for every $\fc'' \ni F \leftrightsquigarrow G \in \gc''$, so $\fc'' = \gc''$.
\end{proof}

The final step is to realize the diagonal matrix $D$ in the statement of 
\autoref{lem:induction}. To do this, we will apply the duality theory developed in
\autoref{sec:duality} to prove the following lemma.
\begin{restatable}{lemma}{intertwinelemma}
  \label{lem:d_intertwine}\RestateRemark
    If $\fc$ and $\gc$ on domain $[q]$ are Holant-indistinguishable, then there exist
    $V, U \in O(q)$ such that, after replacing $\fc$ by
    $V \fc$ and $\gc$ by $U \gc$,
    there is a nonzero diagonal matrix $D \in \rr^{q \times q}$ such that,
    for every corresponding pair $F \in \fc$ and $G \in \gc$ with
    arity $n \geq 2$ and every $m,d \geq 1$ with $m+d = n$, we have 
    \[
        D^{\otimes m} F^{m,d} = G^{m,d} D^{\otimes d}.
    \]
\end{restatable}

While the inductive \autoref{lem:induction} gives an explicit procedure to construct a 
transformation between $\fc$ and $\gc$ given $D$ and the subdomain transformations,
the transformations $V$ and $U$ and the matrix $D$ 
intertwining $\fc$ and $\gc$ are obtained nonconstructively in the proof of
\autoref{lem:d_intertwine}, which applies a black-box invariant-theoretic theorem of 
Schrijver \cite{schrijver_tensor_2008}.

Now we combine the inductive \autoref{lem:induction} and the nonconstructive 
\autoref{lem:d_intertwine} to prove \autoref{thm:result}.
\begin{proof}[Proof of Theorem \ref{thm:result}]
    $(ii) \implies (i)$ is \autoref{cor:holanttheorem}.
    We show $(i) \implies (ii)$. Let $\fc,\gc$ be Holant-indistinguishable.
    We proceed by induction on the domain size $q$. If $q = 1$ then the only orthogonal matrices
    in $\rr^{q \times q}$ are $\pm I$ and every signature $F \in \rr^{[1]^n}$, regardless of arity $n$, 
    has a single entry, which we denote by $F_0 \in \rr$. Let $F \leftrightsquigarrow G$ have
    even arity $n$. Contracting $\frac{n}{2}$ pairs of inputs of both $F$ and $G$, 
    we obtain corresponding signature grids with values $F_0$ and $G_0$, respectively, so $F_0 = G_0$, hence
    $F = G$. For $F \leftrightsquigarrow G$ of odd arity $n$, 
    \[
        F_0^2 = \|F\|^2 = \left\langle F,F\right\rangle = \left\langle G,G\right\rangle = \|G\|^2 = G_0^2,
    \]
    so $F_0 = \pm G_0$. Let $F \leftrightsquigarrow G$ have odd arity $n$
    and $F' \leftrightsquigarrow G'$ have odd arity $n'$, all nonzero. Then
    $F^{\otimes n'} \otimes (F')^{\otimes n}$
    has even arity $2nn'$; contracting its $nn'$ pairs of inputs gives a $\fc$-grid
    with value $F_0^{n'}(F'_0)^n$. The corresponding $\gc$-grid has value
    $G_0^{n'}(G'_0)^n$. Let $G_0 = (-1)^aF_0$ and $G'_0 = (-1)^{a'} F'_0$ for $a,a'\in \{0,1\}$. Then
    \[
        F_0^{n'} (F'_0)^n = G_0^{n'}(G'_0)^n = (-1)^{an' + a'n}F_0^{n'}(F'_0)^n
        = (-1)^{a + a'}F_0^{n'}(F'_0)^n
    \]
    (where in the final equality we used that $n,n'$ are odd), so $a = a'$. Thus there is a common
    $a \in \{0,1\}$ such that $((-1)^a)^n F = G$ for every $n$-ary $F \leftrightsquigarrow G$,
    so $(-I)^a \fc = \gc$.

    Now assume $q > 1$. By \autoref{prop:closure}, we may assume $\fc$ and $\gc$ are 
    quantum-gadget-closed. Let $V,U \in O(q)$ be the transformations given by 
    \autoref{lem:d_intertwine}.
    By \autoref{prop:transform}, we may replace $\fc$ by $V \fc$ and $\gc$ by $U \gc$.
    Now, by \autoref{lem:d_intertwine},
    there is a nonzero diagonal matrix $D \in \rr^{q \times q}$ such that,
    for every corresponding pair $F \in \fc$ and $G \in \gc$ with
    arity $n \geq 2$ and every $m,d \geq 1$ with $m+d = n$,
    \begin{equation}
        D^{\otimes m} F^{m,d} = G^{m,d} D^{\otimes d}.
        \label{eq:d_intertwine}
    \end{equation}
    We consider two cases for $D$: either $D \in \spn(I)$ or $D \not\in \spn(I)$.
    First, suppose $D \in \spn(I)$, so $D = cI$ for $c \neq 0$. 
    Let $F \leftrightsquigarrow G$ be nonzero $n$-ary signatures with $n \geq 2$
    (note that $F = 0 \iff G = 0$ because
    $\|F\|^2 = \langle F,F\rangle = \langle G,G\rangle = \|G\|^2$).
    Applying \eqref{eq:d_intertwine} with $m = n-1$ and $d = 1$ gives
    \begin{equation}
        \label{eq:cfg}
        c^{n-2} F = G.
    \end{equation}
    Then
    \begin{equation}
        \label{eq:fnorm}
        \|F\|^2 = \langle F,F\rangle = \langle G,G\rangle = c^{2(n-2)}\|F\|^2.
    \end{equation}
    As $\fc$ and $\gc$ are quantum-gadget-closed, there are some
    $\fc \ni F \leftrightsquigarrow G \in \gc$ with arity $n \geq 3$, 
    so \eqref{eq:fnorm} gives $c = \pm 1$.
    Now applying \eqref{eq:cfg} to any $n$-ary pair $F \leftrightsquigarrow G$ with $n \geq 2$
    gives $c^nF = c^{n-2}F = G$, so $(cI) F = G$, with $cI \in O(q)$.
    For unary ($n=1$) $F$ and $G$, since $\fc$ and $\gc$ are quantum-gadget-closed, they contain the 
    ternary signatures $F^{\otimes 3}$ and $G^{\otimes 3}$, respectively. So,
    by the previous reasoning, $(cI) F^{\otimes 3} = G^{\otimes 3}$, or equivalently
    $(cF)^{\otimes 3} = G^{\otimes 3}$, which implies $cF = G$, as $F$ and $G$ are 
    real-valued. Combining the non-unary and unary cases gives $(cI) \fc = \gc$.

    Otherwise, $D \not\in \spn(I)$. 
    We will show $\fc \cup \{D\}$ and $\gc \cup \{D\}$ are Holant-indistinguishable,
    then apply the inductive \autoref{lem:induction}. 
    The proof is illustrated in \autoref{fig:grid_trace}.
    \begin{figure}[ht!]
        \centering
\includegraphics[width=1\textwidth]{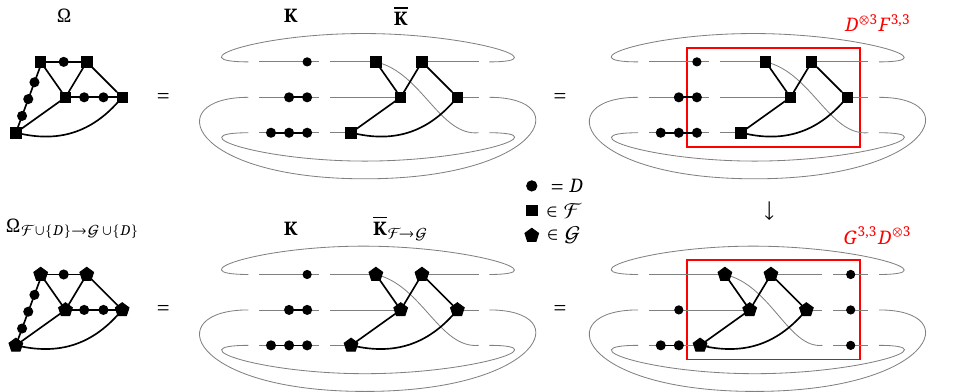}
        \caption{The Holant-value-preserving transformation from $\Omega$ to
        $\Omega_{\fc\cup\{D\}\to\gc\cup\{D\}}$ in the $D \not\in \spn(I)$ case.
        The transition from the bottom right grid to the bottom center grid wraps the
        three $D$ vertices on the right around to the left along their respective wires.}
        \label{fig:grid_trace}
    \end{figure}
    Consider a $\fc \cup \{D\}$-grid $\Omega$ with at least one vertex assigned $D$ (if $\Omega$ has
    no such vertex then we are done, as $\fc$ and $\gc$ are Holant-indistinguishable). Let $\vk \subset
    \Omega$ be the subgadget induced by all vertices assigned $D$.
    Any connected component of $\vk$ is either a cycle or 
    a binary path gadget with signature $D^m$ for some $m$. 
    The multiplicative factors from corresponding $D$-cycles in $\Omega$
    and $\Omega_{\fc\cup\{D\}\to\gc\cup\{D\}}$ cancel, so, disregarding its cycle components,
    $\vk$ consists of $p$ disconnected path gadgets for some $p$. By rearranging the dangling edges
    of $\vk$ and $\overline{\vk}$,
    we may assume $\vk \in \mathfrak{G}_{\{D\}}(p,p)$ with
    $M(\vk) = \bigotimes_{i=1}^p D^{m_i}$ for $m_1,\ldots,m_p \geq 1$, and furthermore that
    $\overline{\vk} \in \mathfrak{G}_{\fc}(p,p)$, and that
    connecting the $i$th left input and $i$th right input of $\vk \circ \overline{\vk}$,
    for $i \in [p]$, reconstructs $\Omega$ (see \autoref{fig:grid_trace}).
    Since $\overline{\vk}$ is an $\fc$-gadget and $\fc$ is
    quantum-gadget-closed, $\overline{\vk}$ has signature $F$ for some $F \in \fc$. Then
    \begin{equation}
        \label{eq:trace_f}
        \holant_\Omega = \tr\left(M(\vk) M(\overline{\vk})\right)
        = \tr\left(\left(\bigotimes_{i=1}^p D^{m_i}\right) F^{p,p}\right).
    \end{equation}
    With $F \leftrightsquigarrow G$, we similarly have
    \begin{equation}
        \label{eq:trace_g}
        \holant_{\Omega_{\fc\cup\{D\}\to\gc\cup\{D\}}}
        = \tr\left(\left(\bigotimes_{i=1}^p D^{m_i}\right) G^{p,p}\right).
    \end{equation}
    Now \eqref{eq:trace_f}, \eqref{eq:d_intertwine} with $m=d=p$, and \eqref{eq:trace_g} 
    give
    \begin{align*}
        \holant_{\Omega}
        &= \tr\left(\left(\bigotimes_{i=1}^p D^{m_i}\right) F^{p,p}\right) \\
        &= \tr\left(\left(\bigotimes_{i=1}^p D^{m_i-1}\right) D^{\otimes p}  F^{p,p}\right) \\
        &= \tr\left(\left(\bigotimes_{i=1}^p D^{m_i-1}\right) G^{p,p} D^{\otimes p}\right) \\
        &= \tr\left(D^{\otimes p}\left(\bigotimes_{i=1}^p D^{m_i-1}\right) G^{p,p} \right) \\
        &= \tr\left(\left(\bigotimes_{i=1}^p D^{m_i}\right) G^{p,p} \right) \\
        &= \holant_{\Omega_{\fc\cup\{D\}\to\gc\cup\{D\}}}.
    \end{align*}
    Thus $\fc \cup \{D\}$ and $\gc \cup \{D\}$ are Holant-indistinguishable. This fact, along with the
    induction hypothesis, lets us apply \autoref{lem:induction} to conclude that $\fc \cup \{D\}$
    and $\gc \cup \{D\}$ are ortho-equivalent. Therefore, by \autoref{prop:union}, $\fc$ and $\gc$
    are ortho-equivalent.
\end{proof} 

\section{Quantum Gadget Duality}
\label{sec:duality}
In this section, we use an invariant-theoretic
result of Schrijver \cite{schrijver_tensor_2008} to prove a
combinatorial-algebraic duality (\autoref{thm:intertwiner_gadget}) showing that
quantum $\fc$-gadgets exactly capture all tensors invariant under the
action of the group of orthogonal transformations stabilizing $\fc$.
The proof of \autoref{thm:intertwiner_gadget} entails characterizing quantum $\fc$-gadgets as 
the algebra generated by a few fundamental gadgets under operations respecting the action 
of $O(q)$. 
\autoref{thm:intertwiner_gadget} and its proof are extensions of Regts' 
\cite[Theorem 3]{regts_rank_2012}
from edge-coloring models to general signature sets $\fc$, largely translated into the language of
signature matrices. We then use \autoref{thm:intertwiner_gadget} to prove the key 
nonconstructive \autoref{lem:d_intertwine}.

We will use the following two
notations from invariant theory \cite{schrijver_tensor_2008, regts_rank_2012}.
\begin{definition}[$T(\rr^q)^Q$, $\stab(\fc)$]
    \label{def:invariance}
    For a subgroup $Q \subset O(q)$, define
    \[
        T(\rr^q)^Q := \{F \in T(\rr^q) \mid HF = F \text{ for every } H \in Q\} \subset T(\rr^q).
    \]
    Dually, for a signature set $\fc \subset T(\rr^q)$, define
    \[
        \stab(\fc) := \{H \in O(q) \mid HF = F \text{ for every } F \in \fc\} \subset O(q).
    \]
\end{definition}
\begin{definition}[$\mathcal{W}$]
    Let $\mathcal{W} \subset T(\rr^q)$ be the signatures of
    gadgets with no vertices -- that is, the signatures of gadgets consisting of only wires and 
    vertexless loops.
\end{definition}
A gadget consisting of $n$ wires has a $2n$-ary
signature in $\mathcal{W}$, which uniquely corresponds, via the order of the gadget's dangling edges,
to a partition (or matching) of $[2n]$ into two-element subsets.
Since the gadgets defining $\mathcal{W}$ have no vertices, they belong to $\gk$ for any $\fc$.
This view of $\mathcal{W}$ as a set of universal signatures is reinforced by the following
classical theorem from representation theory, called the (tensor) \emph{First Fundamental Theorem} for
$O(q)$, proved by Weyl \cite{weyl_classical_1966}, and stated in this
form by Regts \cite[Theorem 4.3]{regts}. 
Define $\langle \fc \rangle_+$ as the set of all $\rr$-linear 
combinations of (matching-arity) signatures in $\fc$.
\begin{theorem}[FFT for $O(q)$]
    $T(\rr^q)^{O(q)} = \langle \mathcal{W} \rangle_+$.
    \label{thm:fft}
\end{theorem}
We will not need the $\subset$ inclusion of \autoref{thm:fft}, although it follows directly from
our \autoref{thm:intertwiner_gadget} below. The $\supset$ inclusion follows from a simple calculation
and has the geometric intuition shown in \autoref{fig:fft}.
\begin{figure}[ht!]
    \centering
\includegraphics[width=0.8\textwidth]{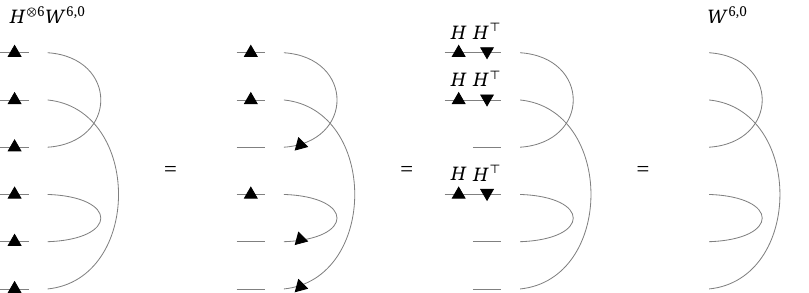}

    \caption{The gadget perspective on the $\supset$ inclusion of \autoref{thm:fft}, showing
        $H^{\otimes 6}W^{6,0} = W^{6,0}$ for 6-ary 
        $W \in \mathcal{W}$ and orthogonal $H$. Each triangle vertex is assigned $H$.
        Reassociating indices in the gadget composition expression has the 
        diagrammatic interpretation of moving each lower
        $H$ along its wire to the top, transposing it in the process
        (its left and right dangling edges are switched). Then each $H H^\top$ resolves to
        the identity.}
    \label{fig:fft}
\end{figure}

Say $\fc \subset T(\rr^q)$ is a \emph{graded subalgebra} of $T(\rr^q)$ if 
$\fc \cap (\rr^q)^{\otimes n}$ (the signatures of arity $n$) 
is a vector space over $\rr$ (closed under $\rr$-linear combinations) and
$\fc$ is closed under $\otimes$.
Say $\fc$ is \emph{contraction-closed} if, for every $n$-ary $F \in \fc$,
the $(n-2)$-ary signature resulting from \emph{contracting} any two inputs of $F$ 
(connecting the corresponding dangling edges of the $n$-ary gadget consisting of a single vertex
assigned $F$) is also in $\fc$.
The following theorem, due to Schrijver, is the key connection between
algebra and combinatorics underlying
our proof of \autoref{thm:intertwiner_gadget} and, eventually, \autoref{thm:result}.
\begin{theorem}[{\cite[{Corollary 1e}]{schrijver_tensor_2008}}]
    Let $\fc \subset T(\rr^q)$. Then there is a subgroup $Q \subset O(q)$ with $\fc = T(\rr^q)^Q$
    if and only if $\fc$ is a contraction-closed graded subalgebra of $T(\rr^q)$ containing $I$.
    \label{thm:duality}
\end{theorem}
Schrijver's statement of \autoref{thm:duality} requires that
$\fc$ be ``nondegenerate'', but it can be seen from the proof that
it suffices to assume that $\fc$ contains $I$ (cf. \cite[Theorem 4]{regts_rank_2012}).

\begin{definition}[$M(\fc),M(T(\rr^q))$]
    Let $\fc \subset T(\rr^q)$ be a signature set. Define 
    \[
        M(\fc) := \bigcup_{F \in \fc}\left(\bigcup_{\substack{m,d \geq 0 \\ m+d=\arity(F)}} 
        F^{m,d}\right)
    \]
    to be the set of all flattenings of all tensors in $\fc$. In particular,
    \[
        M(T(\rr^q)) = \bigcup_{m,d \geq 0} \rr^{q^m\times q^d}.
    \]
\end{definition}
For a subset $P \subset M(T(\rr^q))$, define $\tcwd{P} \subset M(T(\rr^q))$ to be the set generated
by $P$ under $\circ, \otimes, \top$, and $\rr$-linear combinations of matrices with matching
dimensions.

\begin{definition}[$S_\sigma, S$]
    For permutation $\sigma \in S_n$, define the $2n$-ary \emph{braid} signature $S_\sigma \in \mathcal{W}$ 
    by $(S_\sigma)_{\vx\vy} = 1$ iff $x_i = y_{\sigma(i)}$ for every $i \in [n]$, and 
    $(S_\sigma)_{\vx\vy} = 0$ otherwise.

    Define the 4-ary `swap' signature $S := S_{(01)}$. See \autoref{fig:braid}
\end{definition}
\begin{figure}[ht!]
    \centering
\includegraphics[width=0.5\textwidth]{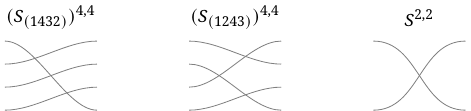}
    \caption{Braid gadgets and their signature matrices.}
    \label{fig:braid}
\end{figure}
Every permutation $\sigma$ is a product of adjacent transpositions, so the flattened braid
signature $(S_\sigma)^{n,n} \in \tcwdn{I,S^{2,2}}$ \cite[Lemma 3]{young2022equality}.
The following is a central object of study in the works of
Man\v{c}inska, Roberson, Cai and Young \cite{planar, cai_planar_2023, young2022equality}. 
\begin{definition}
\label{def:tcwd}
For $C \subset M(T(\rr^q))$, define $C(m,d) := C \cap \rr^{q^m\times q^d}$. Say
$C$ is a \emph{tensor category with duals} (TCWD) if it satisfies the following properties:
\begin{enumerate}[label=(\roman*)]
    \item For every $m,d \geq 0$, $C(m,d)$ is a vector space over $\rr$,
    \item $C$ is closed under $\circ, \otimes, \top$,
    \item $I \in C(1,1)$,
    \item $I^{0,2} \in C(0,2)$.
\end{enumerate}
Say $C$ is a \emph{symmetric} tensor category with duals if it also satisfies
\begin{enumerate}[label=(\roman*),start=5]
    \item $S^{2,2} \in C(2,2)$.
\end{enumerate}
\end{definition}

Towards translating \autoref{thm:duality} into the language of TCWDs,
we next show that a symmetric TCWD is exactly a flattened contraction-closed graded
subalgebra of $T(\rr^q)$.
\begin{lemma}
    \label{lem:tcwd}
    A subset $\fc \subset T(\rr^q)$ is a contraction-closed graded subalgebra of $T(\rr^q)$ containing
    $I$ if and only if $M(\fc)$ is a symmetric tensor category with duals.
\end{lemma}
\begin{proof}
    ($\Longrightarrow$): Let $\fc$ be a contraction-closed graded subalgebra of $T(\rr^q)$ containing $I$. 
    As observed by Schrijver \cite{schrijver_tensor_2008}, 
    for any $n$-ary $F \in \fc$, we can apply any permutation $\sigma \in S_n$ to the inputs
    of $F$ by constructing $F \otimes I^{\otimes n}$, then contracting the $i$th input of $F$ with
    one input of the $\sigma(i)$th copy of $I$.

    First, $I,I^{0,2} \in M(\fc)$ because $I \in \fc$. We can construct $S \in \fc$ by contracting 
    inputs to the
    first and third copies of $I$ and second and fourth copies of $I$ in $I^{\otimes 4} \in \fc$,
    so $S^{2,2} \in M(\fc)$. Item (i) of \autoref{def:tcwd} holds because $\fc$ is graded.
    Finally, we show that $M(\fc)$ satisfies item (ii). See \autoref{fig:ops_contraction} (a),(b),(c),
    respectively.

    \begin{figure}[ht!]
        \begin{center}
        \begin{subfigure}{0.4\textwidth}
            \begin{subfigure}{1\textwidth}
                \centering
\includegraphics[width=1\textwidth]{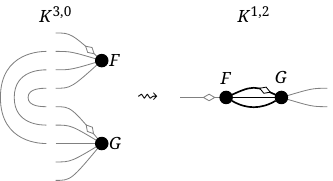}
                \caption*{(a) Constructing $F^{1,3}\circ G^{3,2}$}
            \end{subfigure}

            \vspace{1cm}
            \begin{subfigure}{1\textwidth}
                \centering
\includegraphics[width=1\textwidth]{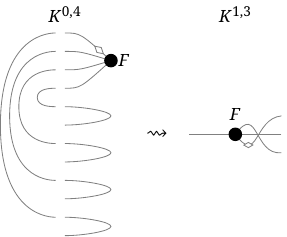}
                \caption*{(c) Constructing $(F^{3,1})^\top$}
            \end{subfigure}
          \end{subfigure}
          \hspace{0.1\textwidth}
        \begin{subfigure}{0.4\textwidth}
            \centering
\includegraphics[width=1\textwidth]{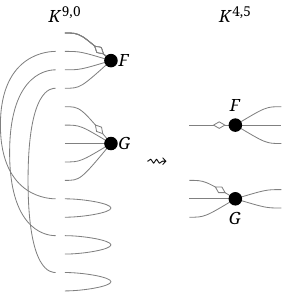}
            \caption*{(b) Constructing $F^{1,3}\otimes G^{3,2}$}
        \end{subfigure}
        \end{center}
        \caption{Realizing signature matrix $\circ,\otimes,\top$ using contractions,
            represented by a wire between two dangling edges. Signature inputs are in
        counterclockwise order, with a diamond indicating the first input.}
        \label{fig:ops_contraction}
    \end{figure}

    \begin{itemize}
        \item $\circ$: Let $F^{m,d}, G^{d,k} \in M(\fc)$. Construct the signature $K \in \fc$
        by starting with $F \otimes G \in \fc$ and
        contracting the $d$ lowest-indexed inputs of $F$ (the right/column inputs in 
        $F^{m,d}$)
       with the $d$ highest-indexed inputs of $G$, in reverse order. Then $F^{m,d}G^{d,k} = K^{m,k} \in
        M(\fc)$.
        \item $\otimes$: Let $F^{m,d}, G^{m',d'} \in M(\fc)$. Construct the signature $K \in \fc$ by starting
        with $F \otimes G \in \fc$ and reordering the inputs of $K$ as follows: the first $m$ inputs
        to $F$, then all inputs to $G$, then the final $d$ inputs to $F$.
        Then $F^{m,d} \otimes G^{m',d'} = K^{m+m',d'+d} \in M(\fc)$.
        \item $\top$: Let $F^{m,d} \in M(\fc)$. Construct the signature $K \in \fc$ by reversing the
            input order of $F \in \fc$. Then $(F^{m,d})^\top = K^{d,m} \in M(\fc)$.
    \end{itemize}

    ($\Longleftarrow$): Let $M(\fc)$ be a symmetric TCWD, so $S,I \in \fc$, and $\fc$ is a
    graded subspace of $T(\rr^q)$ by part (i) of \autoref{def:tcwd}. 
    If $F,G \in \fc$, then their signature vectors $f,g \in M(\fc)$, so 
    $f \otimes g \in M(\fc)$ by closure of $M(\fc)$ under $\otimes$. Then 
    $F \otimes G \in \fc$ because $f \otimes g = (F \otimes G)^{n,0}$ is the signature vector 
    of $F \otimes G$.
    Thus $\fc$ is closed under $\otimes$, so is a graded subalgebra of $T(\rr^q)$.
    Recall that, for any permutation $\sigma \in S_n$, the signature $S_\sigma$ is 
    constructible from $I$ and $S$, so $S_{\sigma}^{n,n} \in M(\fc)$. 
    Let $F \in \fc$ have arity $n$.
    To contract inputs $i$ and $j$ of $F$, define $\sigma \in S_n$ to move inputs $i$
    and $j$ to positions 1 and 2, and shift all other inputs down while maintaining their order.
    Then connect inputs $i$ and $j$ using $I^{0,2}$: 
    \begin{equation}
        (I^{0,2} \otimes I^{\otimes n-2}) \circ S_{\sigma}^{n,n} \circ f
        \in M(\fc)
        \label{eq:contraction}
    \end{equation}
    is the signature vector of the $ij$-contraction of $F$. See \autoref{fig:contraction}.
    \begin{figure}[ht!]
        \centering
\includegraphics[width=0.5\textwidth]{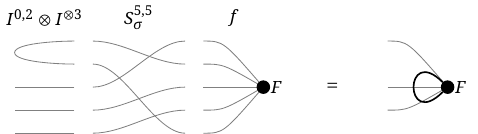}
        \caption{Contracting the 2nd and 5th inputs of $F$ as in \eqref{eq:contraction}.}
        \label{fig:contraction}
    \end{figure}
\end{proof}

Define an action of $O(q)$ on $M(T(\rr^q))$ by 
\[
    HF^{m,d} := H^{\otimes m} F^{m,d} (H^\top)^{\otimes d}.
\]
Using this action, define the matrix version of \autoref{def:invariance}: for $Q \subset O(q)$ and
$\fc \subset T(\rr^q)$,
\begin{align*}
    &M(T(\rr^q))^Q := \{F^{m,d} \in M(T(\rr^q)) \mid H^{\otimes m} F^{m,d} (H^\top)^{\otimes d} = F^{m,d} \text{ for every } H \in Q\} \text{ and}\\
    &\stab(M(\fc)) := \{H \in O(q) \mid H^{\otimes m}F^{m,d} (H^\top)^{\otimes d} = F^{m,d} \text{ for every } F^{m,d} \in M(\fc)\}. 
\end{align*}
The next proposition shows that the matrix and signature versions of these definitions agree.
\begin{proposition}
    \label{prop:invariant_tcwd}
    ~
    \begin{enumerate}
        \item For any subgroup $Q \subset O(q)$, $M(T(\rr^q))^Q$ is a symmetric tensor category with duals.
        \item For any $n$-ary signature $F$ and orthogonal $H$,
            $(HF)^{m,d} = H^{\otimes m} F^{m,d} (H^\top)^{\otimes d}$ for any $m + d = n$.
        \item For any $\fc \subset T(\rr^q)$, $\stab(\fc) = \stab(M(\fc))$. 
    \end{enumerate}
\end{proposition}
\begin{proof}
    If $F^{m,d},G^{d,k} \in M(T(\rr^q))^Q$ and $H \in Q$ then
    \[
        H^{\otimes m} (F^{m,d} G^{d,k}) (H^\top)^{\otimes k} = 
        (H^{\otimes m} F^{m,d} (H^\top)^{\otimes d}) (H^{\otimes d} G^{d,k} (H^\top)^{\otimes k})
        = F^{m,d} G^{d,k},
    \]
    so $F^{m,d}G^{d,k} \in M(T(\rr^q))^Q$. Closure under the other operations $+, \otimes, \top$
    follows similarly, so $M(T(\rr^q))^Q$ satisfies items (i) and (ii) of \autoref{def:tcwd}.
    Since $I \in \mathcal{W}$, \autoref{thm:fft} gives $I \in T(\rr^q)^Q$, or
    $I^{2,0} \in M(T(\rr^q))^Q$. Then $I^{0,2} = (I^{2,0})^\top \in M(T(\rr^q))^Q$,
    and $I = I^{1,1} \in M(T(\rr^q))^Q$ because $HIH^\top = I$ for $H \in Q \subset
    O(q)$, giving items (iii) and (iv). Hence $M(T(\rr^q))^Q$ is a TCWD. 

    We now prove part 2.
    Let $G = HF$. Define the orthogonal block matrix 
    \begin{equation}
        \label{eq:b}
        B := \begin{bmatrix} 0 & H^\top \\ H & 0 \end{bmatrix}
    \end{equation}
    on $V(F) \sqcup V(G) = [2q]$.
    Since $H^{\otimes n} f = g$ (so also $(H^{\top})^{\otimes n}g = f$),
    \autoref{cor:off_diag_block} with $m:=n$ and $d:=0$ gives
    $f \oplus g = (F \oplus G)^{n,0} \in M(T(\rr^{2q}))^{\langle B\rangle}$, which, by the previous paragraph, is a TCWD.
    Now, as in \cite[Lemma 3]{cai_planar_2023}, we use $I$ and $I^{0,2}$ to `pivot' any number of
    inputs of $f \oplus g$ from left to right 
    (which preserves the cyclic dangling edge order) to show 
    $(F \oplus G)^{m,d} \in M(T(\rr^{2q}))^{\langle B \rangle}$ for any $m+d=n$. See \autoref{fig:pivot}.
    \begin{figure}[ht!]
        \centering
\includegraphics[width=0.8\textwidth]{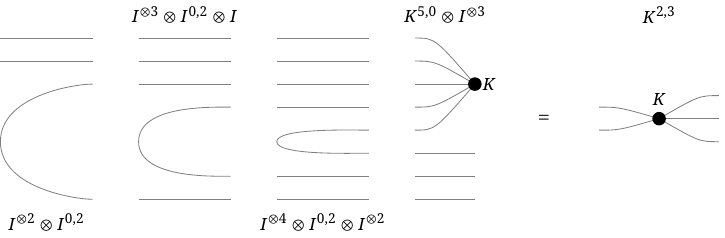}
        \caption{Pivoting three dangling edges of $K = F \oplus G$ from left to right using 
        $I$ and $I^{0,2}$.}
        \label{fig:pivot}
    \end{figure}
    Thus $B^{\otimes m}(F \oplus G)^{m,d} = (F \oplus G)^{m,d} B^{\otimes d}$, so applying the
    other direction of \autoref{cor:off_diag_block} gives
    $H^{\otimes m} F^{m,d} = G^{m,d} H^{\otimes d} = (HF)^{m,d} H^{\otimes d}$, and part 2 follows.

    Part 3 follows directly from part 2 and the definitions. 
    Finally, $S \in T(\rr^q)^Q$ by \autoref{thm:fft} (as $S \in \mathcal{W}$),
    so $Q \subset \stab(S) = \stab(M(\{S\}))$ by part 3. Then
    $S^{2,2} \in M(T(\rr^q))^Q$, giving item (v) of \autoref{def:tcwd} and proving part 1.
\end{proof}
That part 2 of \autoref{prop:invariant_tcwd} holds for any matrix $H$
is a well-known fact (see \autoref{fig:pivot_h}). We gave the above
proof to mirror the proof of the similar \cite[Lemma 6]{cai_planar_2023} for quantum
orthogonal matrices, and to highlight the role of $I^{0,2}$ 
in the invariance of a TCWD under edge pivoting.

Since $I,S \in \mathcal{W} \subset [\fc]$ for any $\fc$, and $\gk$ is closed under gadget 
$\circ,\otimes$, and $\top$, the space $M(\qk) = M([\fc])$ is a symmetric TCWD. Specifically,
the next proposition, which is similar to \cite[Theorem 3]{young2022equality}, shows that $M(\qk)$ is the symmetric TCWD generated by $M(\fc)$.
\begin{proposition}
    $M(\qk) = \tcwd{\big\{I,I^{0,2},S^{2,2}\big\} \cup \big\{f \mid F \in \fc\big\}}$.
    \label{prop:mqf}
\end{proposition}
\begin{proof}
    $I, S\in \mathcal{W} \subset [\fc]$, so $I^{1,1},I^{0,2},S^{2,2} \in M(\qk)$.
    The $\supseteq$ direction follows. For the $\subseteq$ direction, it suffices to show
    that $M(\gk) \subseteq \tcwdn{\big\{I,I^{0,2},S^{2,2}\big\} \cup \big\{f \mid F \in \fc\big\}}$.
    Let $\vk \in \gk(m,d)$ be a gadget. By using $I^{2,0} = (I^{0,2})^\top$ to pivot dangling edges to the
    left as in the proof of \autoref{prop:invariant_tcwd}, we may assume $d = 0$.
    Suppose $\vk$ has $p$ vertices, assigned signatures $F_1,\ldots,F_{p}$.
    Break all internal edges of $\vk$ and orient the resulting dangling edges to the left 
    to create a gadget with signature
    \begin{equation}
        \bigotimes_{i=1}^{p} f_i \in \tcwdn{f \mid F \in \fc}.
        \label{eq:biggadget}
    \end{equation}
    Now we reconstruct $\vk$ by reconnecting (contracting) the broken internal edges. 
    Using the formula \eqref{eq:contraction} for arbitrary contraction, we see that
    $\vk \in \tcwd{\big\{I,I^{0,2},S^{2,2}\big\} \cup \big\{f \mid F \in \fc\big\}}$.
\end{proof}

\begin{proposition}
    \label{prop:stabfc}
    $\stab(\fc) = \stab(M(\qk)) = \stab([\fc])$.
\end{proposition}
\begin{proof}
    By part 3 of \autoref{prop:invariant_tcwd}, $\stab(M(\qk)) = \stab(M([\fc])) = \stab([\fc])$.
    Since $M(\qk)$ contains (the vector forms of) $\fc$, we have
    $\stab(M(\qk)) \subseteq \stab(\fc)$. By \autoref{prop:mqf}, 
    $M(\qk)$ is the symmetric TCWD generated by $\fc$ (in particular, is the smallest symmetric
    TCWD containing $\fc$).
    By part 1 of \autoref{prop:invariant_tcwd}, $M(T(\rr^q))^{\stab(\fc)}$ is a symmetric TCWD
    containing $\fc$, so we must have $M(\qk) \subseteq M(T(\rr^q))^{\stab(\fc)}$, which implies
    $\stab(\fc) \subseteq \stab(M(\qk))$.
\end{proof}

Now we come to the main result of this section. We follow Regts' proof
for the special case of edge coloring models \cite[Theorem 6.8]{regts}.
\begin{theorem}
    \label{thm:intertwiner_gadget}
    Let $\fc$ be a set of real-valued signatures on domain $[q]$. Then
    \[
        T(\rr^q)^{\stab(\fc)} = [\fc].
    \]
    Equivalently, by flattening both sides and applying part 3 of \autoref{prop:invariant_tcwd},
    we have, for any $m,d$,
    \[
        \left(\rr^{q^m\times q^d}\right)^{\stab(\fc)} = M(\qk(m,d)).
    \]
\end{theorem}
\begin{proof}
    By \autoref{prop:mqf}, $M([\fc]) = M(\qk)$ is a TCWD and
    $S \in [\fc]$, so $[\fc]$ is a contraction-closed
    graded subalgebra of $T(\rr^q)$ containing $I$ by \autoref{lem:tcwd}. Thus, by \autoref{thm:duality},
    there is a subgroup $G \subseteq O_q(\rr)$ such that $[\fc] = T(\rr^q)^G$.
    This gives $G \subseteq \stab([\fc])$, so $T(\rr^q)^{\stab([\fc])} \subseteq T(\rr^q)^G$, and 
    also $T(\rr^q)^G = [\fc] \subseteq T(\rr^q)^{\stab([\fc])}$. Therefore,
    applying \autoref{prop:stabfc},
    \[
        T(\rr^q)^{\stab(\fc)} = T(\rr^q)^{\stab([\fc])} = T(\rr^q)^G = [\fc].
        \qedhere
    \]
\end{proof}
A direct consequence is the FFT for $O(q)$. Schrijver \cite{schrijver_tensor_2008} gives another 
proof based on \autoref{thm:duality}.
\begin{proof}[Proof of Theorem \ref{thm:fft}]
    Putting $\fc = \varnothing$ (the empty set) in \autoref{thm:intertwiner_gadget} gives
    \[
        T(\rr^q)^{O(q)} = T(\rr^q)^{\stab(\varnothing)} = [\varnothing] = \langle \mathcal{W} \rangle_+,
    \]
    as $\varnothing$-gadgets must have no vertices.
\end{proof}

\begin{remark}
    \label{rem:duality}
    The reasons we translated from 
    Schrijver and Regts's \cite{schrijver_tensor_2008, regts_rank_2012} language of
    contraction-closed graded subalgebras to
    Man\v{c}inska, Roberson, Cai and Young
    \cite{planar, cai_planar_2023, young2022equality}'s language of symmetric TCWDs are twofold. 
    The first is to highlight the similarities
    between this section's results and the results of the latter group. 
    Enforcing that
    the TCWDs contain $M(\eq)$ (corresponding to quantum gadgets in the context of $\csp$), 
    replacing $O(q)$ with $S_q$ (the symmetric group of permutation matrices) or $S^+_q$ 
    (the \emph{quantum symmetric group}), 
    and dropping the symmetry condition on the TCWDs 
    in the latter case (corresponding to quantum gadget planarity), we find 
    \autoref{thm:duality} analogous to the
    \emph{Tannaka-Krein duality} used by \cite{young2022equality} and 
    \cite{planar, cai_planar_2023}, respectively, to prove results analogous to 
    \autoref{thm:intertwiner_gadget}. In fact, the results of this section show that
    \autoref{thm:duality} is equivalent to a well-known classical version of Tannaka-Krein
    duality \cite[Theorem 1.3]{banica_liberation_2009}.

    The second reason is that, just as the main result of \cite{cai_planar_2023} is a planar, quantum
    version
    of the main result of \cite{young2022equality}, our \autoref{thm:result} should have a planar, 
    quantum
    version (see \autoref{con:quantum} below). Arbitrary contractions do not respect planarity, but,
    by simply removing the symmetry condition (i.e. removing $S^{2,2}$ in \autoref{prop:mqf}), 
    we easily enforce planarity in the language of TCWDs.
\end{remark}

We now have the tools necessary to prove \autoref{prop:closure}, which we restate here.
\closurelemma*
\begin{proof}
    1: Any $\fc$-grid or $\gc$-grid is also a $[\fc]$-grid or $[\gc]$-grid, respectively, giving
    the $(\Longleftarrow)$ direction. 
    For $(\Longrightarrow)$, we can express any $[\fc]$-grid $\Omega$
    as a quantum $\fc$-grid by, for every vertex $v$ in $\Omega$ assigned a signature $F^v \in
    [\fc] \setminus \fc$, replacing the subgadget of $\Omega$ induced by $v$ by the quantum
    $\fc$-gadget with signature $F^v$, then linearly expanding to obtain a quantum $\fc$-grid.
    Do the same for $\Omega_{[\fc]\to[\gc]}$. 
    By assumption, the resulting corresponding quantum $\fc$ 
    and $\gc$-grids have the same values, so $\holant_\Omega = \holant_{\Omega_{[\fc]\to[\gc]}}$.

    2: $(\Longleftarrow)$ is direct from $\fc \subset [\fc]$ and $\gc \subset [\gc]$. For $(\Longrightarrow)$, suppose
    $H\fc = \gc$ for orthogonal $H$. Consider $\fc \oplus \gc$ and the block matrix 
    $B$ in \eqref{eq:b}. As in the proof of \autoref{prop:invariant_tcwd},
    $B \in \stab(\fc \oplus \gc)$. Then, by \autoref{prop:stabfc}, $B \in 
    \stab([\fc \oplus \gc])$. Now consider $n$-ary 
    $[\fc] \ni F' \leftrightsquigarrow G' \in [\gc]$,
    where $F'$ is the signature of $\vk \in \qk(n,0)$, and $G'$ is the signature of $\vk_{\fc\to\gc}$.
    We may assume no gadget in $\vk$ has a component without dangling edges
    by absorbing the factors from such components (which must be equal for
    $\vk$ and $\vk_{\fc\to\gc}$ by \autoref{cor:holanttheorem}) into the gadget's coefficient,
    so, by definition of $\oplus$, $M(\vk_{\fc\to\fc\oplus\gc})|_{V(\fc)^n} = M(\vk) = f'$ and
    $M(\vk_{\fc\to\fc\oplus\gc})|_{V(\gc)^n} = M(\vk_{\fc\to\gc}) = g'$ (other blocks are not necessarily
    0 if $\vk$ is still disconnected). Then, setting
    $K^{n,0} := M(\vk_{\fc\to\fc\oplus\gc}) \in M([\fc \oplus \gc])$ and considering 
    the equation $K^{n,0} = B^{\otimes n}K^{n,0}$ in block form as in \eqref{eq:hmkmd}, we have
    \[
        \begin{bmatrix} 
            f' \\ * \\ \vdots \\ * \\ g'
        \end{bmatrix} =
        \begin{bmatrix} 
            0 & 0 & \ldots & 0 &(H^\top)^{\otimes n}\\
            0 & 0 & \ldots & * & 0\\
            \vdots & \vdots & \iddots & \vdots & \vdots \\
            0 & * & \ldots & 0 & 0\\
            H^{\otimes n} & 0 & \ldots & 0 & 0
        \end{bmatrix}
        \begin{bmatrix} 
            f' \\ * \\ \vdots \\ * \\ g'
        \end{bmatrix}
        =
        \begin{bmatrix} 
            (H^\top)^{\otimes n} g' \\ * \\ \vdots \\ * \\ H^{\otimes n} f'
        \end{bmatrix}.
    \]
    In particular, $g' = H^{\otimes n}f'$. Therefore $H[\fc] = [\gc]$.
\end{proof}

Next, we apply the full power of duality (\autoref{thm:intertwiner_gadget}) to prove the
nonconstructive part of \autoref{lem:d_intertwine}.
\begin{lemma} \label{lem:offdiagblock}
    If $\fc$ and $\gc$ are Holant-indistinguishable, then there is a transformation
    $H \in \stab(\fc \oplus \gc)$ satisfying $H|_{V(\fc),V(\gc)} \neq 0$ or 
    $H|_{V(\gc),V(\fc)} \neq 0$.
\end{lemma}
\begin{proof}
    Suppose every $H \in \stab(\fc \oplus \gc)$ has $H|_{V(\fc),V(\gc)} = H|_{V(\gc),V(\fc)} = 0$
    (i.e. is block diagonal). Then the block diagonal
    matrix $A = \begin{bmatrix} I & 0 \\ 0 & 2I \end{bmatrix}$
    satisfies $H A = A H$, hence $HAH^T = A$, for every $H \in \stab(\fc \oplus \gc)$. Thus,
    by \autoref{thm:intertwiner_gadget} with $m=d=1$, $A$ is realizable as the signature matrix 
    of a binary quantum gadget
    -- that is, there exist binary ($\fc \oplus \gc$)-gadgets $\vk^1, \ldots \vk^p$ and
    $c_1,\ldots,c_p \in \rr$ such that
    \begin{equation}
        \label{eq:a_quantum_gadget}
        A = \sum_{i=1}^p c_i M(\vk^i).
    \end{equation}
    Any connected component of a gadget $\vk^i$ disconnected from the component(s) of $\vk^i$ containing the
    two dangling edges contributes only an overall multiplicative factor; by absorbing this factor
    into $c_i$, we may assume each $\vk^i$ has no components without a dangling edge. By definition
    of $\oplus$, inputting an
    $x \in V(\fc)$ along a dangling edge of $\vk^i$ forces any edge assignment with nonzero value to 
    assign
    an element of $V(\fc)$ to every edge in that dangling edge's connected component. So, for any $x,y \in V(\fc)$,
    we have $M(\vk^i)_{x,y} = M(\vk^i_{(\fc \oplus \gc) \to \fc})_{x,y}$. 
    Similar reasoning applies to
    $\gc$, so the $V(\fc),V(\fc)$ and $V(\gc),V(\gc)$ blocks of \eqref{eq:a_quantum_gadget} are
    \begin{equation}
        I = \sum_{i=1}^p c_i M(\vk^i_{(\fc \oplus \gc) \to \fc})
        \text{ and }
        2I = \sum_{i=1}^p c_i M(\vk^i_{(\fc \oplus \gc) \to \gc}),
        \label{eq:zero_and_i}
    \end{equation}
    respectively.
    Let $\Omega^i$ be the $\fc$-grid resulting from connecting the two dangling edges of 
    $\vk^i_{(\fc \oplus \gc) \to \fc}$. Then
    taking the trace of the equations in \eqref{eq:zero_and_i} gives
    \[
        \sum_{i=1}^p c_i \holant_{\Omega^i}(\fc) = q \neq 2q = 
        \sum_{i=1}^p c_i \holant_{\Omega^i_{\fc\to\gc}}(\gc),
    \]
    so there is some $i$ for which $\holant_{\Omega^i}(\fc) \neq \holant_{\Omega^i_{\fc\to\gc}}(\gc)$.
\end{proof}

Finally, we apply some postprocessing to \autoref{lem:offdiagblock} to
prove \autoref{lem:d_intertwine}, which we restate here.
\intertwinelemma*
\begin{proof}
    By \autoref{lem:offdiagblock}, there is an $H \in \stab(\fc\oplus\gc)$
    with either $H|_{\fc,\gc} \neq 0$ or $H|_{\gc,\fc} \neq 0$. Assume WLOG that $H|_{\gc,\fc} \neq 0$.
    Let $H|_{\gc,\fc} = U^\top D V$ be the singular value decomposition of $H|_{\gc,\fc}$, with
    $U,V$ orthogonal and $D \neq 0$ diagonal, all real. 
    Replace $\fc$ with $V \fc$ and $\gc$ with $U \gc$.
    This has the effect of replacing $\fc \oplus \gc$ by
    $(V\fc)\oplus (U \gc) = (V \oplus U)(\fc \oplus \gc)$ (by \eqref{eq:oplus_action}),
    which replaces
    $\stab(\fc\oplus\gc)$ by $(V \oplus U) \circ \stab(\fc\oplus\gc) \circ 
    (V \oplus U)^{-1} = (V \oplus U) \circ \stab(\fc\oplus\gc) \circ 
    (V \oplus U)^\top$. In particular, $H$ is replaced by
    \[
        \begin{bmatrix} V & 0 \\ 0 & U \end{bmatrix} 
        \begin{bmatrix} H|_{\fc,\fc} & H|_{\fc,\gc} \\ U^\top DV & H|_{\gc,\gc} \end{bmatrix}
        \begin{bmatrix} V^\top & 0 \\ 0 & U^\top \end{bmatrix} =
        \begin{bmatrix} VH|_{\fc,\fc}V^\top & VH|_{\fc,\gc}U^\top \\ U(U^\top DV) V^\top & U H|_{\gc,\gc} U^\top \end{bmatrix}
        = \begin{bmatrix} * & * \\ D & * \end{bmatrix}.
    \]
    To summarize, after transforming $\fc$ by $V$ and $\gc$ by $U$, we have
    $H = \begin{bmatrix} * & * \\ D & * \end{bmatrix} \in \stab(\fc \oplus \gc)$ for nonzero diagonal
    $D$. 
    Let $F \in \fc$ and $G \in \gc$ be corresponding signatures of arity $n \geq 2$. 
    For any $m,d \geq 1$ with $m+d = n$, part 3 of \autoref{prop:invariant_tcwd} gives
    \begin{equation}
        \label{eq:hfg}
        H^{\otimes m}(F \oplus G)^{m,d} = (F \oplus G)^{m,d}H^{\otimes d}.
    \end{equation}
    Now, by \autoref{prop:block} with $K := F \oplus G$ and \eqref{eq:oplus_index},
    we can write \eqref{eq:hfg} as the block matrix equation
    \begin{equation}
        \begin{bmatrix} 
            * & * & \ldots & *\\
            \vdots & \vdots & \iddots & \vdots\\
            * & * & \ldots & *\\
            D^{\otimes m} & * & \ldots & *
        \end{bmatrix}
        \begin{bmatrix} 
            F^{m,d} & 0 & \ldots & 0 & 0\\
            0 & 0 &\ldots & 0 & 0\\
            \vdots & \vdots & \ddots & \vdots & \vdots\\
            0 & 0 &\ldots & 0 & 0\\
            0 & 0 & \ldots & 0 & G^{m,d}
        \end{bmatrix}
        =
         \begin{bmatrix} 
            F^{m,d} & 0 & \ldots & 0 & 0\\
            0 & 0 &\ldots & 0 & 0\\
            \vdots & \vdots & \ddots & \vdots & \vdots\\
            0 & 0 &\ldots & 0 & 0\\
            0 & 0 & \ldots & 0 & G^{m,d}
        \end{bmatrix}           
        \begin{bmatrix} 
            * & * & \ldots & * \\
            \vdots & \vdots & \iddots & \vdots \\
            * & * & \ldots & *\\
            D^{\otimes d} & * & \ldots & *
        \end{bmatrix}
        \label{eq:dfg_blocks}
    \end{equation}
    (note the importance of $m,d \geq 1$ so that $F^{m,d}$ and $G^{m,d}$ are not in the same
    block-column or block-row of the block matrices in \eqref{eq:dfg_blocks}).
    The bottom left block of \eqref{eq:dfg_blocks} is
    $D^{\otimes m} F^{m,d} = G^{m,d} D^{\otimes d}$ as desired.
\end{proof}

\section{Consequences of \autoref{thm:result}}
\label{sec:corollaries}
In this section, we exploit the expressiveness of the Holant framework to show that \autoref{thm:result}
encompasses a variety of existing results, and derive a few novel consequences.
\subsection{Counting CSP and graph homomorphisms}
Recall the set $\eq$ of equality signatures from \autoref{sec:holant}.
For signature set $\fc$, define the \emph{counting constraint satisfaction problem} $\csp(\fc)$
with \emph{constraint function} set $\fc$
to be the problem $\holant(\fc \mid \eq)$. Vertices assigned signatures in $\fc$ and $\eq$ are
\emph{constraints} and \emph{variables}, respectively, and a $(\fc\mid\eq)$-grid $\Omega$ is a 
constraint-variable incidence graph, where a variable appears in all of its incident constraints (the equality signatures force each variable to take the same value in all constraints in
which it appears). Then
$\holant_\Omega(\fc\mid\eq)$ is the sum over all variable assignments of the product of the constraint
evaluations.
Like Holant, $\csp$ is a well-studied
problem in counting complexity, with dichotomy theorems classifying $\csp(\fc)$ as either tractable
or \#P-hard proved for increasingly broad classes of constraint function sets 
\cite{bulatov_2013, dyer_richerby, cai-chen-lu, cai-chen-complexity}. 

By inserting a dummy degree-2 constraint vertex assigned $(=_2) \in \eq$ between adjacent variable vertices and combining adjacent constraint
vertices assigned $=_a$ and $=_b$ into a single constraint vertex assigned $=_{a+b-2}$, we see that
$\holant(\fc \cup \eq)$ is equivalent to $\holant(\fc \mid \eq)$. See \autoref{fig:bipartify}.
\begin{figure}[ht!]
    \centering
\includegraphics[width=0.6\textwidth]{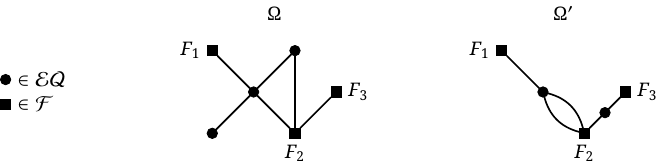}
    \caption{A $\holant(\fc \cup \eq)$ grid $\Omega$ and the equivalent 
    $\holant(\fc \mid \eq)$ grid $\Omega'$.}
    \label{fig:bipartify}
\end{figure}
Say $\fc$ and $\gc$ are \emph{isomorphic} if there exists a permutation
matrix $H$ satisfying $H\fc = \gc$ (in other words, $\fc$ and $\gc$
are the same up to relabeling of their domains). 
Say that $\fc$ and $\gc$ are $\csp$-indistinguishable if $\fc\cup\eq$ and $\gc\cup\eq$ are 
Holant-indistinguishable (in other words, every $\csp$ instance has the same value whether we
use constraint functions from $\fc$ or from $\gc$). If $\fc$ and $\gc$ are isomorphic, then
they are $\csp$-indistinguishable.
Applying \autoref{thm:result} to $\fc \cup \eq$ and $\gc \cup \eq$, we obtain the converse
of this statement, the main result of Young \cite{young2022equality} for 
real-valued constraint functions.
\begin{corollary}
    Let $\fc$ and $\gc$ be sets of real-valued constraint functions.
    Then $\fc$ and $\gc$ are isomorphic if and only if $\fc$ and $\gc$ are $\csp$-indistinguishable.
    \label{cor:isomorphism}
\end{corollary}
\begin{proof}
    We only need $(\Longleftarrow)$.
    If $\fc$ and $\gc$ are $\csp$-indistinguishable, then $\fc \cup \eq$ and $\gc \cup \eq$
    are Holant-indistinguishable, so, by \autoref{thm:result}, there
    is an orthogonal $H$ such that $H\fc = \gc$ and $H \eq = \eq$.
    Using a standard Vandermonde
    interpolation argument, Xia \cite{xia} shows that any matrix $H$ satisfies $H \eq = \eq$ 
    if and only $H$ is a permutation matrix. Therefore $\fc$ and $\gc$ are isomorphic.
\end{proof}
Let $\eq_2 \subset \eq$ be the set of equality signatures of even arity and
$\csp^2(\fc) := \holant(\fc \mid \eq_2)$ be $\csp(\fc)$
restricted to instances in which every variable appears an even number of times \cite{cai2015holant,huang_2016_dichotomy}. We now obtain a sharpening of \autoref{cor:isomorphism}. A
\emph{signed permutation matrix} is an orthogonal matrix with entries
in $\{0,\pm 1\}$ and exactly one nonzero entry in each row and column.
\begin{corollary}\label{cor:signed}
    Let $\fc$ and $\gc$ be sets of real-valued constraint functions.
    Then there is a signed permutation matrix $P$ satisfying $\gc = P\fc$ if and only if
    $\fc$ and $\gc$ are $\csp^2$-indistinguishable.
\end{corollary}
\begin{proof}
    As above, $\csp^2(\fc) \equiv \holant(\fc \mid \eq_2)$ is equivalent to
    $\holant(\fc \cup \eq_2)$
    (critically, if $=_a \in \eq_2$ and $=_b \in \eq_2$, then $=_{a+b-2} \in \eq_2$). 
    Schrijver \cite{schrijver_tensor_2008} 
    shows\footnote{The First Fundamental Theorem for $S_q^{\pm} \subset O(q)$ (the group of signed
        permutation matrices) states that $T(\rr^q)^{S_q^{\pm}} = [\eq_2]$
        (cf. \autoref{thm:fft}, and recall that $\langle \mathcal{W} \rangle_+ = [\varnothing]$).
        It follows as in the proof of \autoref{thm:intertwiner_gadget} that $\stab(\eq_2) = S_q^{\pm}$.
        The fact that $\stab(\eq) = S_q \subset O(q)$ (the group of permutation matrices) similarly follows
        from the First Fundamental Theorem for $S_q$, which states that
        $T(\rr^q)^{S_q} = [\eq]$.
    }
    that any invertible matrix $H$ satisfies $H\eq_2 = \eq_2$ if and only if $H$ is a signed 
    permutation matrix. 
    Now, as in the proof of \autoref{cor:isomorphism}, the result follows from applying
    \autoref{thm:result} to $\fc \cup \eq_2$ and $\gc \cup \eq_2$.
\end{proof}
As discussed around \autoref{fig:hom}, $\holant(A_X \mid \eq) \equiv \csp(A_X)$ counts the 
number of homomorphisms to graph $X$ from multigraphs $K$. Therefore it follows from
\autoref{cor:isomorphism} that two graphs $X$ and $Y$ are isomorphic if and only if they
admit the same number of homomorphisms from every multigraph $K$. Furthermore, it follows 
directly from the definition that the number of homomorphisms from a multigraph $K$ to graph 
$X$ equals the number of homomorphisms 
from $K'$ to $X$, where $K'$ is the graph constructed from $K$ by
removing all loops and duplicate edges. Hence \autoref{cor:isomorphism}
is a generalization of the classical theorem of Lovász 
\cite{lovasz_operations} that two graphs are isomorphic if and only if they admit the same 
number of homomorphisms from every graph. Similarly, \autoref{cor:signed} characterizes
homomorphism indistinguishability from multigraphs in which all vertices have even degree.
However, this gives no new information about homomorphisms from graphs, as
duplicating all edges of any graph $K'$ yields a multigraph with 
only even-degree vertices. Nevertheless, by a slightly more involved application of 
\autoref{cor:signed}, we obtain the following sharpening of Lovász's theorem (which also
follows from a \emph{connector} argument in the style of \cite{dvorak_recognizing_2010}).
\begin{corollary}
    Graphs $X$ and $Y$ are isomorphic if and only if they admit the same number of homomorphisms
    from those graphs in which all vertices have even degree.
\end{corollary}
\begin{proof}
    We only need $(\Longleftarrow)$. Consider
    $\holant(A_X^3 \mid \eq_2) \equiv \csp^2(A_X^3)$ and
    $\holant(A_Y^3 \mid \eq_2) \equiv \csp^2(A_Y^3)$ (where $A_X^3$ is the third power of
    matrix $A_X$). Let $\Omega$ be a
    $\holant(A_X^3 \mid \eq_2)$ grid, and $K_\Omega$ be the multigraph underlying $\Omega$ in
    the sense of \autoref{fig:hom} (ignoring the degree-2 vertices assigned $A_X^3)$.
    Now replace every vertex in $\Omega$ assigned $A_X^3$ with a path of five vertices 
    assigned $A_X$, $=_2$, $A_X$, $=_2$, $A_X$
    (a gadget with signature matrix $A_X I A_X I A_X = A_X^3$).
    The Holant value of the resulting 
    $\holant(A_X \mid \eq_2)$ grid $\Omega'$ equals the
    Holant value of $\Omega$, and, following \autoref{fig:hom}, equals the number of 
    homomorphisms to $X$ from a graph (not multigraph) $K_{\Omega'}$ in which all vertices 
    have even 
    degree, as $K_{\Omega'}$ is constructed from $K_\Omega$ -- a multigraph in which all 
    vertices have even degree -- by subdividing every edge twice to eliminate all loops
    and parallel edges. Applying the assumption and similar reasoning to
    the $Y$ side, we conclude that $A_X^3$ and $A_Y^3$ are
    $\csp^2$-indistinguishable. Thus \autoref{cor:signed} gives a signed
    permutation matrix $P'$ satisfying $P' A_X^3 (P')^\top = A_Y^3$.
    Since $A_X^3$ and $A_Y^3$ have nonnegative entries and the action of $P'$ only rearranges
    and possibly negates the entries of $A_X^3$, there must in fact be no negation, so
    $P A_X^3 P^\top = A_Y^3$,
    where $P$ is the permutation matrix created by flipping every $-1$ entry of $P'$ to $1$. 
    Since $A_X^3$ and $A_Y^3$ are diagonalizable with real eigenvalues (in particular, the
    eigenvalues have unique cube roots), it follows that 
    $P A_X P^\top = A_Y$, so $X$ and $Y$ are isomorphic.
\end{proof}

\subsection{Simultaneous matrix similarity}
Let $\fc$ and $\gc$ be sets of binary signatures, thought of as matrices. Any connected $\fc$-grid
$\Omega$
is a cycle. Breaking an edge of the cycle, we obtain a binary path gadget with signature matrix
$\prod_{i=1}^c F_i$, where, depending on its orientation, each $F_i \in \fc$ or $F_i^\top \in \fc$.
Connecting the path's two dangling edges, we reform $\Omega$, which thus has Holant value
$\tr\left(\prod_{i=1}^c F_i\right)$. Let $\Gamma_{\fc}$ be the set of all finite products of
matrices in $\fc$ and $\fc^\top := \{F^\top \mid F \in \fc\}$. Define $\Gamma_{\gc}$ similarly and,
for a word $w \in \Gamma_{\fc}$, construct $w_{\fc\to\gc} \in \Gamma_{\gc}$ by replacing every character
$F$ or $F^\top$ in $w$ by the corresponding $G$ or $G^\top$, respectively. 
For orthogonal $H$, we have $H\fc = \gc \iff HF^{1,1} = G^{1,1}H$ for every $F \leftrightsquigarrow G$
(by part 2 of \autoref{prop:invariant_tcwd}), so, in this
setting, \autoref{thm:result} is equivalent to the following real-valued case of a
classical theorem from representation theory, due to Specht \cite{specht} and Wiegmann \cite{wiegmann}.
Grohe, Rattan, and Seppelt \cite{grohe_homomorphism_2022} also give a combinatorial proof.
\begin{corollary}
    Let $\fc, \gc \subset \rr^{q\times q}$ be similar. 
    Then there is an $H \in O(q)$ such that $HF = GH$ 
    for every corresponding $F \in \fc$ and $G \in \gc$ if and only if
    $\tr(w) = \tr(w_{\fc\to\gc})$ for every $w \in \Gamma_{\fc}$.
    \label{cor:specht}
\end{corollary}

Suppose $\fc = \{A_X\}$ and $\gc = \{A_Y\}$ for graphs $X$ and $Y$.
Transform an
$A_X$-grid $\Omega$ to a $(A_X \mid \eq)$-grid $\Omega'$ by inserting a dummy degree-2 vertex assigned 
$I = (=_2) \in \eq$ between every consecutive pair of vertices in the cycle. Recall from
\autoref{sec:holant} that
$\holant_{\Omega'}(A_X \mid \eq)$ counts the number of homomorphisms from graph $K$ to $X$, where $K$ is the graph obtained from $\Omega'$ by ignoring the vertices assigned $A_X$. Here
$K$ is a cycle, so we have the following well-known result, an alternate formulation
of this case of \autoref{cor:specht}.
\begin{corollary}
    Let $X$ and $Y$ be graphs. Then there is an orthogonal matrix $H$ satisfying $H A_X = A_Y H$
    if and only if $X$ and $Y$ admit the same number of homomorphisms from all cycles.
    \label{cor:cycle}
\end{corollary}

A matrix $H$ is \emph{pseudo-stochastic} if all of its rows and columns sum to 1. Dell, Grohe, and
Rattan \cite{dell} proved that graphs $X$ and $Y$ admit the same number of homomorphisms from all
paths if and only if there is a pseudo-stochastic matrix $H$ such that $HA_X = A_YH$. 
Using \autoref{thm:result}, we combine this result with \autoref{cor:cycle}, which also
reproduces a combinatorial explanation 
for the connection between pseudo-stochastic matrices and homomorphisms from paths 
\cite{grohe_homomorphism_2022}.
\begin{corollary}
    Let $X$ and $Y$ be graphs. Then there is a pseudo-stochastic orthogonal matrix $H$ satisfying 
    $H A_X = A_Y H$
    if and only if $X$ and $Y$ admit the same number of homomorphisms from all cycles and paths.
\end{corollary}
\begin{proof}
    Consider $\holant(A_X \cup \{=_1\})$. Any $A_X\cup \{=_1\}$-grid is a disjoint union of 
    cycles composed of signatures assigned $A_X$ and
    paths with degree-2 internal vertices assigned $A_X$ and degree-1 endpoints
    assigned $=_1 \in \eq$. As discussed before
    \autoref{cor:cycle}, every cycle $A_X$-grid $\Omega$ has the same Holant value as
    $\Omega_{A_X \to A_Y}$ if and only if $X$ and $Y$ admit the same
    number of homomorphisms from every cycle. Similarly inserting dummy vertices assigned $=_2$ between
    every pair of $A_X$ vertices in a path component, we produce
    a $(A_X \mid \eq)$-grid whose Holant value equals the number of homomorphisms to $X$ from the
    underlying path.
    Thus $X$ and $Y$ admit the same number of homomorphisms from all cycles and all paths
    if and only if $A_X \cup \{=_1\}$ and $A_Y \cup \{=_1\}$ are Holant-indistinguishable.
    By \autoref{thm:result}, this is equivalent to the existence of an orthogonal $H$ satisfying 
    $H A_X = A_Y H$ and $H (=_1) = (=_1)$.
    The vector form of $=_1$ is the all-ones vector, so $H (=_1) = (=_1)$ if and only if the rows
    of $H$ sum to $1$ and (since $H(=_1) = (=_1) \iff H^\top(=_1) = (=_1)$) the columns of $H$
    sum to 1.
\end{proof}

\subsection{Odeco signature sets}
In this section, we derive a consequence of \autoref{thm:result} that is not a counting
indistinguishability theorem,
but a combinatorial characterization of signatures that are simultaneously `diagonalizable'.
\begin{definition}[$\geneq$, odeco]
    Define the set of \emph{general equalities} (or \emph{weighted equalities}) on domain $[q]$ as
    $\geneq = \{=_n^{\va} \mid n \geq 0, \va \in \rr^{[q]}\}$, where $=_n^{\va}$ is the symmetric
    $n$-ary signature defined by
    \[
        (=_n^{\va})_{\vx} = \begin{cases} a_i & x_1 = \ldots = x_n = i \\ 0 & \text{otherwise}.
        \end{cases}
    \]
    for $\vx \in [q]^n$.
    A set $\fc$ of symmetric signatures is \emph{orthogonally decomposable}, or \emph{odeco}, if it is ortho-equivalent
    to a general equality set -- that is, there exists an $H \in O(q)$ such that $H\fc \subset \geneq$.
\end{definition}

The term ``odeco'' was coined by Robeva \cite{robeva} to refer to individual symmetric tensors
(signatures) which are ortho-equivalent to a general equality. A binary $\geneq$ signature has a
diagonal signature matrix, so the spectral theorem states that every (real) symmetric binary 
signature is odeco
(recall part 2 of \autoref{prop:invariant_tcwd}).
Any nonzero edge assignment for a connected $\geneq$-gadget $\vk$ must assign all edges, including 
dangling edges, the same domain element, so, if $\vk$ has arity $n$ and is composed of vertices assigned signatures with 
weights $\va^1,\ldots,\va^p$, then $\vk$ has signature $=_n^{\va^1 \bullet \ldots \bullet \va^p} \in \geneq$, where $\bullet$ denotes entrywise product.
In particular, if $\vk$ is a $\geneq$-grid $\Omega$, then
$\holant_\Omega = \sum_{i=1}^q (\va^1 \bullet \ldots \bullet \va^p)_i$. Thus, if 
$\fc$ is odeco, then the computational problem $\holant(\fc)$ is polynomial-time tractable via
orthogonal Holographic transformation.
\begin{definition}[$*$]
    For symmetric signatures $F_1,F_2 \in \fc$ of arity $n_1$ and $n_2$, respectively, 
    construct the $(n_1+n_2-2)$-ary signature $F_1 * F_2 \in [\fc]$ from $F_1 \otimes F_2$
    by contracting an input of $F_1$ and an input of $F_2$.
    See \autoref{fig:star}.
\end{definition}
\begin{figure}[ht!]
    \centering
\includegraphics[width=0.2\textwidth]{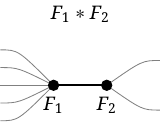}
    \caption{Illustrating (the gadget with signature) $F_1 * F_2$ for 6-ary $F_1$ and 3-ary $F_2$.}
    \label{fig:star}
\end{figure}
Since $F_1$ and $F_2$ are symmetric,
$F_1 * F_2$ doesn't depend on which inputs we connect.
For $\vx \in [q]^{n_1-1}$ and $\vy \in [q]^{n_2-1}$, we have (with vectors viewed as input lists)
\[
    (F_1 * F_2)(\vx,\vy) = \sum_{z \in [q]} F_1(\vx,z) F_2(\vy,z).
\]
\begin{proposition}
    \label{prop:odeco}
    For any $H \in O(q)$ and signature $F \in (\rr^q)^{\otimes n}$ of arity $n \geq 3$, 
    we have $HF \in \geneq$ if and only if $H(F * F) \in \geneq$. 
\end{proposition}
\begin{proof}
    $(\Longrightarrow)$: If $HF = E$ for $H \in O(q)$ and $E \in \geneq$, then by part 2 of 
    \autoref{prop:closure}, $H(F * F) = (HF) * (HF) = E * E \in \geneq$.

    $(\Longleftarrow)$: 
    Let $\geneq \ni E = H(F * F) = (HF) * (HF)$ for $H \in O(q)$.
    Suppose toward contradiction that $HF \not\in
    \geneq$, so there is a $\vx \in [q]^n$ such that $(HF)(\vx) \neq 0$ but $\exists i,j$ such that 
    $x_i \neq x_j$. Assume WLOG that $i,j \neq n$, as $n \geq 3$. 
    Construct $\vx' \in [q]^{n-1}$ by deleting the $n$th (last) entry of $\vx$. Then
    \[
        E(\vx',\vx') = ((HF) * (HF))(\vx',\vx') = \sum_{z \in [q]} (HF)(\vx',z)^2 \geq (HF)(\vx)^2 > 0,
    \]
    contradicting $E \in \geneq$. Thus $HF \in \geneq$.
\end{proof}
\begin{theorem}
    \label{thm:odeco}
    Let $\fc$ be a set of real-valued symmetric signatures (tensors). The following are equivalent.
    \begin{enumerate}[label=(\roman*)]
        \item $\fc$ is odeco.
        \item Every connected $\fc$-gadget has a symmetric signature.
        \item For every $F_1,F_2 \in \fc$, the signature $F_1 * F_2$ is symmetric.
    \end{enumerate}
\end{theorem}
Robeva \cite{robeva} conjectured the equivalence of items (i) and (iii) when $\fc$
contains a single signature.
Boralevi, Draisma, Horobeţ, and Robeva \cite{boralevi_orthogonal_2017} confirmed this conjecture
using techniques from algebraic geometry.
We use \autoref{thm:result} to give a combinatorial proof, generalized to arbitrary symmetric
signature sets.
\begin{remark}
    If $\fc$ is a set of symmetric binary signatures, then $F_1 * F_2 = F_1 \circ F_2 = F_1F_2$
    (a matrix product) for $F_1,F_2 \in \fc$. In general, symmetric matrices commute if and only
    if their product is symmetric (as if $F_1F_2$ is symmetric then $F_1F_2 = (F_1F_2)^\top = F_2^\top
    F_1^\top = F_2F_1$ and if $F_1$ and $F_2$ commute then $(F_1F_2)^\top = F_2^\top F_1^\top = F_2F_1 = F_1F_2$).
    Therefore \autoref{thm:odeco} encompasses the extension of the spectral theorem
    which states that commuting symmetric matrices are simultaneously diagonalizable. We use this
    fact in the proof below.
    \label{rem:spectral}
\end{remark}
\begin{proof}[Proof of Theorem \ref{thm:odeco}]
    (i) $\implies$ (ii),(iii): Suppose $H \fc \subset \geneq$ for some $H \in O(q)$.
    Let $K \in [\fc]$ be the signature of a connected $\fc$-gadget (e.g. $K = F_1 * F_2$).
    By part 2 of \autoref{prop:closure}, $HK = J$, where $J$ is the signature of a connected 
    $\geneq$-gadget. Then $J \in \geneq$, so $J$, and hence $K = H^{-1}J$, are symmetric.

    (ii) $\implies$ (i): First, replace every non-unary odd-arity $F \in \fc$ by $F * F$
    (every unary signature is in $\geneq$, so simply remove all unaries from $\fc$). 
    This does not change
    the fact that $\fc$ satisfies item (ii), and, by \autoref{prop:odeco}, does not change whether
    $\fc$ satisfies item (i). Thus we may assume all signatures in $\fc$ have even arity. 
    For $F \in \fc$, let $\widetilde{F}$ be the 
    matrix of the binary signature constructed by contracting all but one pair of inputs of $
    F$ (see \autoref{fig:tilde}).
    \begin{figure}[ht!]
        \centering
\includegraphics[width=0.1\textwidth]{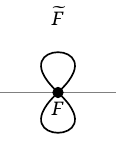}
        \caption{Illustrating (the gadget with signature) $\widetilde{F}$ for 6-ary $F$.}
        \label{fig:tilde}
    \end{figure}
    Since $F$ is symmetric, $\widetilde{F}$ doesn't depend on how we pair up $F$'s inputs. 
    Every $\widetilde{F}$, and every composition $\widetilde{F_1} \circ \widetilde{F_2}$ for
    $F_1,F_2 \in \fc$, is the signature of a connected $\fc$-gadget, so is symmetric by assumption. 
    Therefore, as in \autoref{rem:spectral}, the matrices $\widetilde{F}$ for $F \in \fc$ all commute.
    \begin{claim}
        \label{claim:even}
        If $\vk$ is a connected binary $\fc$-gadget with $p$ vertices, assigned signatures 
        $F_1,\ldots,F_p \in \fc$, then $M(\vk) = \prod_{i=1}^p \widetilde{F_i}$.
      \end{claim}
      \begin{subproof}
            We prove \autoref{claim:even} by induction on $p$. For $p=1$, by the symmetry of $F_1$,
    every connected binary $F_1$-gadget with a single vertex has signature $\widetilde{F_1}$. 
    Now suppose $p \geq 2$. $\vk$ contains a vertex $v$ whose removal does not disconnect
    $\vk$ (take e.g. the final vertex visited by breadth first search). 
    Assume WLOG that $v$ is assigned
    signature $F_p$. Construct $\vk'$ from $\vk$ by breaking 
    all but one edge between $v$ and other vertices (see \autoref{fig:symmetric_gadget}).
    Each broken edge becomes two new dangling edges.
    Since $\vk'$ is connected, its signature is symmetric by assumption.
    The number of dangling edges incident to $v$ is odd,
    as $v$ has even degree and exactly one edge to another vertex (loops on
    $v$ do not affect the parity). Since $\vk'$ has an
    even number of dangling edges (two plus twice the number of edges broken), there are an odd
    number of dangling edges incident to the other vertices of $\vk'$. Now create a binary gadget
    $\vk''$ from $\vk'$ by arbitrarily pairing up and connecting all but one dangling edge incident
    to $v$, and similarly pairing up and connecting all but one dangling edge incident to the other
    vertices of $\vk'$. We may also recover $\vk$ from $\vk'$ by connecting possibly different pairs of
    dangling edges (reforming the edges broken to create $\vk'$) and,
    since the signature of $\vk'$ is symmetric, the signature of a gadget produced by connecting 
    dangling edges of $\vk'$ does not depend on which pairs of dangling edges we connect 
    (although the underlying graphs of the gadgets differ). Therefore $M(\vk'') = M(\vk)$.
    \begin{figure}[ht!]
        \centering
\includegraphics[width=1\textwidth]{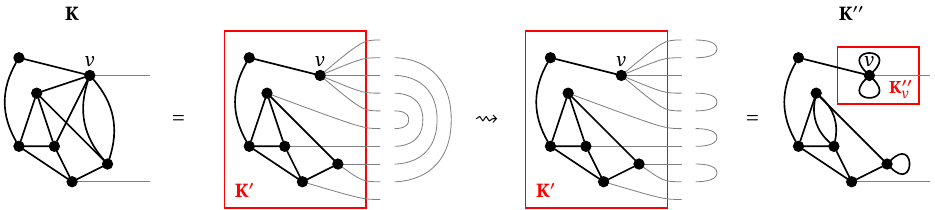}
        \caption{Illustrating the proof of \autoref{claim:even} when $v$ has a single incident dangling
        edge in $\vk$. The cases where $v$ has zero or two incident dangling edges are similar.}
        \label{fig:symmetric_gadget}
    \end{figure}

    Let $\vk''_v$ be the subgadget of $\vk''$ induced by $v$. $\vk''_v$ has
    two dangling edges, one dangling in $\vk''$, and one which, in $\vk''$, connects $v$ to a 
    vertex in $\overline{\vk''_v}$ (the complement of $\vk''_v$), with the remaining edges incident
    to $v$ paired into loops. Thus $M(\vk''_v) = \widetilde{F_p}$. 
    Similarly, $\overline{\vk''_v}$ has two dangling edges, one dangling in $\vk''$ and one incident to
    to $v$ in $\vk''$, so we recover $\vk''$ by reconnecting the edge between $\vk''_v$ and 
    $\overline{\vk''_v}$. Thus, applying induction to $\overline{\vk''_v}$, which has $p-1$ vertices,
    we obtain
    \[
        M(\vk) = M(\vk'') = M(\vk''_v) \circ M(\overline{\vk''_v}) = \widetilde{F_p} \circ
        \prod_{i=1}^{p-1} \widetilde{F_i} = \prod_{i=1}^{p} \widetilde{F_i}
    \]
    (recall that the $\widetilde{F_i}$ commute). This completes the proof of \autoref{claim:even}.
      \end{subproof}

    The matrices $\widetilde{F}$ for $F \in \fc$ are symmetric and commute, so they are simultaneously diagonalizable:
    there is an $H \in O(q)$ such that, for every $F \in \fc$, $H\widetilde{F}H^\top = (=_2^{\va^F})$ 
    for some $\va^F \in \rr^{q}$.
    By \autoref{prop:transform} and part 2 of \autoref{prop:invariant_tcwd},
    we may replace $\fc$ with $H\fc$ to assume each $\widetilde{F} = (=_2^{\va^F})$. Define
    \[
       \gc = \{=_{\arity(F)}^{\va^F} \mid F \in \fc\} \subset \geneq,
    \]
    a set similar to $\fc$.
    Let $\Omega$ be a $\fc$-grid,
    containing signatures $F_1,\ldots,F_p$.
    As $\Omega$ is not a tree (all of its vertices have even degree), we can break some edge of $\Omega$
    to produce a connected binary $\fc$-gadget $\vk$. By \autoref{claim:even},
    \[
        M(\vk) = \prod_{i=1}^{p} \widetilde{F_i} = \prod_{i=1}^{p}(=_2^{\va^{F_i}})
        = \left(=_2^{\va^{F_1} \bullet \ldots \bullet \va^{F_p}}\right) = M(\vk_{\fc\to\gc}).
    \]
    Now, reconnecting the dangling edges of $\vk$, we find $\holant_{\Omega} = 
    \holant_{\Omega_{\fc\to\gc}}$. Thus $\fc$ and $\gc$ are Holant-indistinguishable, so, by
    \autoref{thm:result}, $\fc$ and $\gc$ are ortho-equivalent. Hence $\fc$ is odeco.

    (iii) $\implies$ (ii): Let $K$ be the $n$-ary signature of a connected $\fc$-gadget $\vk$.
    Every unary signature is trivially symmetric, so assume $n \geq 2$. It suffices to show that,
    for any fixed partial input $\vz \in [q]^{n-2}$ to $K$ and any $x,y \in [q]$, we have
    $K(x,y,\vz) = K(y,x,\vz)$ (where we assume WLOG that $x$ and $y$ are the first two inputs to $K$
    by reordering the dangling edges of $\vk$). Let $u$ and $w$ be the vertices of $\vk$ incident to
    the first and second dangling edges of $\vk$ (after reordering).
    If $u = w$ then we are done, as every $F \in \fc$ is symmetric.
    Otherwise, since $\vk$ is connected, it contains a path 
    $P = (u = v_0,v_1,\ldots,v_{p-2},v_{p-1} = w)$ 
    from $u$ to $w$, where $v_i$ is assigned signature $F_i \in \fc$, for $i \in [p]$.
    Let $E(P) := \{e_0,e_1,\ldots,e_{p-1},e_p\}$ be the edges of $P$, including the dangling edges
    $e_0$ and $e_p$ incident to $u$ and $w$, respectively.
    Then $e_i$ and $e_{i+1}$ are inputs to $F_i$ for all $i \in [p]$. 
    For any fixed assignment $\sigma: E(\vk) \setminus E(P) \to [q]$,
    define the matrix $F^\sigma_i \in \rr^{q \times q}$ by 
    $(F^\sigma_i)_{a,b} := F_i(\sigma|_{\delta(v_i)},a,b)$ 
    (that is, fix the inputs to $F_i$ from edges outside of $E(P)$). Then
    \begin{equation}
        K(x,y,\vz) = 
        \sum_{\substack{\sigma: E(\vk) \setminus E(P) \to [q] \\ \sigma(D) = \vz}} 
        \left(\prod_{v \in V(\vk) \setminus P} F_v(\sigma|_{\delta(v)})\right) 
        F^\sigma_P(x,y),
        \label{eq:kxyz}
    \end{equation}
    where $D$ is the ordered list of the last $n-2$ dangling edges of $\vk$ and
    \begin{align*}
        F^\sigma_P(x,y) 
        &= \sum_{\substack{\phi: E(P) \to [q] \\ \phi(e_0)=x,\phi(e_p)=y}} 
            \prod_{i=0}^{p-1} F_i(\sigma|_{\delta(v_i)}, \phi(e_i),\phi(e_{i+1})) \\
        &= \sum_{\substack{\phi: E(P) \to [q] \\ \phi(e_0)=x,\phi(e_p)=y}} 
        \prod_{i=0}^{p-1} (F^\sigma_i)_{\phi(e_i),\phi(e_{i+1})}
        = \left(\prod_{i=0}^{p-1} F^\sigma_i\right)_{x,y}.
    \end{align*}
    On the RHS of \eqref{eq:kxyz}, $x$ and $y$
    appear only in $F_P^\sigma(x,y)$. Thus it suffices to show that, for any fixed $\sigma$,
    $F^\sigma_P(x,y) = F^\sigma_P(y,x)$. For any $i,j \in [p]$ and $a,b \in [q]$,
    \begin{align*}
        &(F_i^\sigma F_j^\sigma)_{a,b} = \sum_{z \in [q]} (F_i^\sigma)_{a,z} (F_j^\sigma)_{z,b}
        = \sum_{z \in [q]} F_i(\sigma|_{\delta(v_i)}, a,z) F_j(\sigma|_{\delta(v_j)}, b,z)
        = (F_i * F_j)(\sigma|_{\delta(v_i)}, a, \sigma|_{\delta(v_j)}, b) \\
        &= (F_i * F_j)(\sigma|_{\delta(v_i)}, b, \sigma|_{\delta(v_j)}, a)
        = \sum_{z \in [q]} F_i(\sigma|_{\delta(v_i)}, b,z) F_j(\sigma|_{\delta(v_j)}, a,z) 
        = \sum_{z \in [q]} (F_i^\sigma)_{b,z} (F_j^\sigma)_{z,a}
        = (F_i^\sigma F_j^\sigma)_{b,a},
        \label{eq:swap}
    \end{align*}
    where the fourth equality uses the assumption that $F_i * F_j$ is symmetric.
    Thus $F_i^\sigma F_j^\sigma$ is symmetric. Both $F_i^\sigma$ 
    and $F_j^\sigma$ are symmetric, as $F_i$ and $F_j$ are symmetric, so, as in \autoref{rem:spectral},
    $F^\sigma_i$ and $F^\sigma_j$ commute. Therefore
    \[
        F^\sigma_P(x,y)
        = \left(\prod_{i=0}^{p-1} F^\sigma_i\right)_{x,y}
        = \left(\left(\prod_{i=0}^{p-1} F^\sigma_i\right)^\top\right)_{y,x}
        = \left(\prod_{i=0}^{p-1} (F^\sigma_{p-1-i})^\top\right)_{y,x}
        = \left(\prod_{i=0}^{p-1} F^\sigma_i\right)_{y,x}
        = F^\sigma_P(y,x). \qedhere
    \]
\end{proof}

\section{Possible Variations of \autoref{thm:result}}
\label{sec:variations}
\subsection{Complex-valued signatures}
\label{sec:complex}
Recall from \autoref{rem:complex} that \autoref{thm:result} does not hold for complex-valued
$\fc$ and $\gc$, as observed by Draisma and Regts \cite{draisma_tensor_2013}.
Draisma and Regts also note that a direct extension of Schrijver's duality 
(\autoref{thm:duality}) to the complex
orthogonal group $O_q(\cc)$ is impossible because $O_q(\cc)$ is not compact. However, they
also provide some techniques for handling complex-valued
signatures in edge coloring models, including a version \cite[Theorem 3]{draisma_tensor_2013} of \autoref{thm:intertwiner_gadget}
for edge-coloring models parameterized by complex-valued sets $\fc$ that are
\emph{nondegenerate}, meaning that the natural bilinear form on the complex tensor
algebra $T(\cc^q)$ is nondegenerate on the subalgebra $[\fc]$ of quantum gadget signatures.

Cai and Young \cite{cai_planar_2023,young2022equality} prove their $\csp$
indistinguishability theorems for complex-valued signature sets with the assumption that the sets
are \emph{conjugate-closed}, meaning they must contain the entrywise conjugate of each of 
their complex signatures (the counterexample set $\fc = \{[1,i]\}$ from \autoref{rem:complex} 
is not conjugate-closed). In particular, conjugate-closed $\fc$ are nondegenerate.
After the publication of the first version \cite{orthogonal} of
this work, Cai and Young \cite{bipartite} showed that 
\autoref{thm:result} indeed holds for complex orthogonal $H$ and nondegenerate
(or \emph{quantum-nonvanishing}) $\fc$ and $\gc$. They achieved this generalization by
applying an extension, due to Derksen and Makam \cite{prop}, of 
\autoref{thm:duality} to nondegenerate complex-valued sets,
leading to a similar extension of \autoref{thm:intertwiner_gadget}.

\subsection{Quantum orthogonal matrices and planar signature grids}
A \emph{quantum orthogonal matrix} $U$ is, roughly, a matrix whose entries come from an abstract,
not necessarily commutative,
$C^*$-algebra, satisfying the relation $UU^\top = U^\top U = I \otimes \mathbf{1}$, where $\mathbf{1}$
is the identity element of the $C^*$-algebra.
Just as a permutation matrix is an orthogonal matrix that stabilizes $\eq$, 
a \emph{quantum permutation matrix} is a quantum orthogonal matrix that
stabilizes $\eq$ (see e.g. \cite[Equation 27]{cai_planar_2023}). The main theorem of Cai and Young
\cite{cai_planar_2023}, extending the result of Man\v{c}inska and Roberson \cite{planar}, is a planar,
quantum version of \autoref{cor:isomorphism}:
$\fc$ and $\gc$ are $\plholant(\cdot \cup \eq)$-indistinguishable 
(planar-\#CSP-indistinguishable) if and only if there is a quantum permutation matrix $U$ satisfying
$U \fc = \gc$ ($\fc$ and $\gc$ are \emph{quantum isomorphic}). 
Removing $\eq$, we should obtain the following planar, quantum version of \autoref{thm:result}.
\begin{conjecture}
    \label{con:quantum}
    Let $\fc$, $\gc$ be sets of real-valued signatures. Then the following are equivalent.
    \begin{enumerate}[label=(\roman*)]
        \item $\plholant_{\Omega}(\fc) = \plholant_{\Omega_{\fc\to\gc}}(\gc)$
        for every planar $\fc$-grid $\Omega$.
        \item There is a quantum orthogonal matrix $U$ such that $U\fc = \gc$.
    \end{enumerate}
\end{conjecture}
Cai and Young \cite[Theorem 5]{cai_planar_2023} prove (ii) $\implies$ (i) when $U$ is a quantum permutation matrix;
however, their proof only relies on $U$ being a quantum orthogonal matrix. Therefore,
to prove \autoref{con:quantum}, it suffices to show (i) $\implies$ (ii). The \emph{Tannaka-Krein duality}
for the quantum symmetric group used by Man\v{c}inska and Roberson and Cai and Young 
\cite[Theorem 3]{cai_planar_2023} has a more general version
for the quantum orthogonal group $O_q^+$ \cite[Theorem 2.13]{planar}. This version is
analogous to \autoref{thm:duality} above, but concerning quantum subgroups of $O_q^+$ and asymmetric
(planar) tensor categories with duals. Then, following the proof of \cite[Theorem 4]{cai_planar_2023},
but ommitting the gadgets $E^{1,0}$ and $E^{1,2}$ used to construct $\eq$, we will obtain a quantum analogue
of our \autoref{thm:intertwiner_gadget} for planar quantum $\fc$-gadgets, giving a quantum
analogue of \autoref{lem:offdiagblock}. The rest of the proof of \autoref{thm:result}, however, involves
nonplanar signature grid manipulations (e.g. in \autoref{fig:grid_trace} it is in general impossible to embed $\Omega$ such that 
every instance of $D$ lies on the outer face) and, more critically, relies on the existence of the
singular value decomposition of a submatrix of a real orthogonal matrix, then on viewing the resulting
diagonal matrix as a signature. It is yet unclear whether the same or similar reasoning applies to a
submatrix of a quantum orthogonal matrix. However, after publication of the first version
\cite{orthogonal} of this work, \autoref{con:quantum} was confirmed using different
techniques in the author's Ph.D.~thesis \cite{young_thesis}.


\subsection*{Acknowledgements}
The author thanks Jin-Yi Cai for helpful discussions and Ashwin Maran for his suggestion which
improved the proof of \autoref{thm:odeco}. The author also thanks Guus Regts for pointing out
reference \cite{schrijver_graph_2008}. Finally, the author thanks the reviewers for their
comments that improved the presentation of the paper and for pointing out an issue
with a previous version of \autoref{prop:odeco}.

\printbibliography

\newpage
\appendix
\section{Block Signature Actions}
\label{sec:appendix_block}
In this appendix, we prove several technical results which state that the action of a 
block matrix $H$ on a block signature $K$ follows block matrix multiplication rules as one would expect.
\begin{proposition}
    \label{prop:block}
    Let $\ic = X \sqcup Y$.
    For $H \in \rr^{\ic\times\ic}$, $K^{m,d} \in \rr^{\ic^m \times \ic^d}$ and any $\vr\in\{X,Y\}^m$
    $\vc\in\{X,Y\}^d$,
    \[
        (H^{\otimes m} K^{m,d})|_{\vr,\vc} = 
        \sum_{\vj \in \{X,Y\}^m} \left(\bigotimes_{i=1}^m H|_{R_i,J_i}\right) K^{m,d}|_{\vj,\vc}
    \]
    (with $H^{\otimes m} K^{m,d} \in \rr^{\ic^m \times \ic^d}$ 
    indexed as in part 3 of \autoref{def:block}) and similarly
    \[
        (K^{m,d}H^{\otimes d})|_{\vr,\vc} = 
        \sum_{\vj \in \{X,Y\}^d}  K^{m,d}|_{\vr,\vj}\left(\bigotimes_{i=1}^d H|_{J_i,C_i}\right).
    \]
\end{proposition}
That is, with $H^{\otimes m}|_{\vr,\vj} = \bigotimes_{i=1}^m H|_{R_i,J_i}$,
we can compute $H^{\otimes m}K^{m,d}$ as a block matrix product with
\[
    H^{\otimes m} = 
    \begin{bmatrix} 
        H^{\otimes m}|_{X^{m},X^{m}} &
        H^{\otimes m}|_{X^{m},X^{m-1}Y} & 
        H^{\otimes m}|_{X^{m},X^{m-2}YX} & \ldots & 
        H^{\otimes m}|_{X^{m},Y^{m}}\\
        H^{\otimes m}|_{X^{m-1}Y,X^{m}} &
        H^{\otimes m}|_{X^{m-1}Y,X^{m-1}Y} & 
        H^{\otimes m}|_{X^{m-1}Y,X^{m-2}YX} & \ldots & 
        H^{\otimes m}|_{X^{m-1}Y,Y^{m}}\\
        \vdots & \vdots & \vdots && \vdots \\
        H^{\otimes m}|_{Y^{m-1}X,X^{m}} & 
        H^{\otimes m}|_{Y^{m-1}X,X^{m-1}Y} &
        H^{\otimes m}|_{Y^{m-1}X,X^{m-2}YX} & \ldots & 
        H^{\otimes m}|_{Y^{m-1}X,Y^{m}}
        \\
        H^{\otimes m}|_{Y^{m},X^{m}} 
        & H^{\otimes m}|_{Y^{m},X^{m-1}Y} 
        & H^{\otimes m}|_{Y^{m},X^{m-2}YX} 
        & \ldots & H^{\otimes m}|_{Y^{m},Y^{m}} 
    \end{bmatrix}
\]
and
\[
    K^{m,d} =
    \begin{bmatrix} 
        K^{m,d}|_{X^{m},X^{d}} &
        K^{m,d}|_{X^{m},X^{d-1}Y} & \ldots &
        K^{m,d}|_{X^{m},Y^{d-1}X} & 
        K^{m,d}|_{X^{m},Y^{d}}\\
        K^{m,d}|_{X^{m-1}Y,X^{d}} &
        K^{m,d}|_{X^{m-1}Y,X^{d-1}Y} & \ldots & 
        K^{m,d}|_{X^{m-1}Y,Y^{d-1}X} & 
        K^{m,d}|_{X^{m-1}Y,Y^{d}}\\
        K^{m,d}|_{X^{m-2}YX,X^{d}} & 
        K^{m,d}|_{X^{m-2}YX,X^{d-1}Y} &\ldots & 
        K^{m,d}|_{X^{m-2}YX,Y^{d-1}X} & 
        K^{m,d}|_{X^{m-2}YX,Y^{d}} \\
        \vdots & \vdots && \vdots & \vdots \\
        K^{m,d}|_{Y^{m},X^{d}} &
        K^{m,d}|_{Y^{m},X^{d-1}Y} & \ldots & 
        K^{m,d}|_{Y^{m},Y^{d-1}X} &
        K^{m,d}|_{Y^{m},Y^{d}}.
    \end{bmatrix}.
\]
\begin{proof}
    We prove the first statement. The second is proved similarly.
    Let $\vrr \in \vr$ and $\vcc \in \vc$. Then
    \begin{align*}
        ((H^{\otimes m} K^{m,d})|_{\vr,\vc})_{\vrr,\vcc}
        = (H^{\otimes m} K^{m,d})_{\vrr,\vcc}
        &= \sum_{\vjj\in\ic^m} \left(\prod_{i=1}^m H_{r_i,j_i}\right)K^{m,d}_{\vjj,\vcc}\\
        &= \sum_{\vj\in\{X,Y\}^m} \sum_{\vjj\in\vj} \left(\prod_{i=1}^m H_{r_i,j_i}\right)K^{m,d}_{\vjj,\vcc}\\
        &= \sum_{\vj \in \{X,Y\}^m} \sum_{\vjj\in\vj}\left(\bigotimes_{i=1}^m H|_{R_i,J_i}\right)_{\vrr,\vjj} (K^{m,d}|_{\vj,\vc})_{\vjj,\vcc}\\
        &= \sum_{\vj \in \{X,Y\}^m} \left(\left(\bigotimes_{i=1}^m H|_{R_i,J_i}\right) K^{m,d}|_{\vj,\vc}\right)_{\vrr,\vcc} \\
        &= \left(\sum_{\vj \in \{X,Y\}^m} \left(\bigotimes_{i=1}^m H|_{R_i,J_i}\right) K^{m,d}|_{\vj,\vc}\right)_{\vrr,\vcc}. \qedhere
    \end{align*}
\end{proof}
We will only need \autoref{prop:block} as written, for partitions of $\ic$ into two blocks, but it is
not hard to see that it extends naturally to partitions of $\ic$ into more than two blocks.

We will apply \autoref{prop:block} for two special types of $H$.
\begin{corollary}
    \label{cor:on_diag_block}
    If 
    \[
        H = H_X \oplus H_Y = \begin{bmatrix} H_X & 0 \\ 0 & H_Y\end{bmatrix}
    \]
    is block-diagonal, then, for any $K \in \rr^{\ic^n}$, 
    the block form of $H^{\otimes n} K^{n,0}$ is
    \[
        \begin{bmatrix} 
            H_X^{\otimes n} & 0 & \ldots & 0 & 0\\
            0 & * & \ldots & 0 & 0\\
            \vdots & \vdots & \ddots & \vdots & \vdots \\
            0 & 0 & \ldots & * & 0\\
            0 & 0 & \ldots & 0 & H_Y^{\otimes n}
        \end{bmatrix}
        \begin{bmatrix} 
            (K|_X)^{n,0} \\ * \\ \vdots \\ * \\ (K|_Y)^{n,0}
        \end{bmatrix} =
        \begin{bmatrix} 
            H_X^{\otimes n}(K|_X)^{n,0} \\ * \\ \vdots \\ * \\ H_Y^{\otimes n}(K|_Y)^{n,0}
        \end{bmatrix}.
    \]
\end{corollary}
\begin{proof}
    Here,
    $H|_{R_i,J_i} = H_{R_i}$ if $R_i = J_i$ and $H|_{R_i,J_i} = 0$ if $R_i\neq J_i$, so $H$ takes the
    claimed block form in \autoref{prop:block}.
\end{proof}
If $X = V(F)$, $Y = V(G)$, and $K = F \oplus G$, then by \eqref{eq:oplus_index},
all blocks of $K$ are 0 except $K|_X = F$ and $K|_Y = G$, so, by \autoref{cor:on_diag_block},
\begin{equation}
    \label{eq:oplus_action}
    (H_{V(F)} \oplus H_{V(G)})(F \oplus G) = (H_{V(F)} F) \oplus (H_{V(G)} G).
\end{equation}

\begin{corollary}
\label{cor:off_diag_block}
If
\[
    B = \begin{bmatrix} 0 & H^\top \\ H & 0 \end{bmatrix}
\]
is block-antidiagonal, with blocks indexed by $X = V(F)$ and $Y = V(G)$, then, for any $m,d$,
\[
    H^{\otimes m}F^{m,d} = G^{m,d}H^{\otimes d} \text{~~and~~}
    (H^\top)^{\otimes m}G^{m,d} = F^{m,d}(H^\top)^{\otimes d}
    \iff 
    B^{\otimes m}(F \oplus G)^{m,d} = (F \oplus G)^{m,d}B^{\otimes d}.
\]
\end{corollary}
\begin{proof}
    $B|_{R_i,J_i} = 0$ if $R_i = J_i$, so, by \autoref{prop:block} and \eqref{eq:oplus_index}, 
    we compute blockwise with
    \begin{equation}
        B^{\otimes m}(F \oplus G)^{m,d} \text{ as }
        \begin{bmatrix} 
            0 & 0 & \ldots & 0 & (H^\top)^{\otimes m}\\
            0 & 0 & \ldots & * & 0\\
            \vdots & \vdots & \iddots & \vdots & \vdots \\
            0 & * & \ldots & 0 & 0\\
            H^{\otimes m} & 0 & \ldots & 0 & 0
        \end{bmatrix}
        \begin{bmatrix} 
            F^{m,d} & 0 & \ldots & 0 & 0\\
            0 & 0 &\ldots & 0 & 0\\
            \vdots & \vdots & \ddots & \vdots & \vdots\\
            0 & 0 &\ldots & 0 & 0\\
            0 & 0 & \ldots & 0 & G^{m,d}
        \end{bmatrix}
        \label{eq:hmkmd}
    \end{equation}
    and
    \begin{equation}
        (F \oplus G)^{m,d} B^{\otimes d} \text{ as }
        \begin{bmatrix} 
            F^{m,d} & 0 & \ldots & 0 & 0\\
            0 & 0 &\ldots & 0 & 0\\
            \vdots & \vdots & \ddots & \vdots & \vdots\\
            0 & 0 &\ldots & 0 & 0\\
            0 & 0 & \ldots & 0 & G^{m,d}
        \end{bmatrix}
        \begin{bmatrix} 
            0 & 0 & \ldots & 0 & (H^\top)^{\otimes d}\\
            0 & 0 & \ldots & * & 0\\
            \vdots & \vdots & \iddots & \vdots & \vdots \\
            0 & * & \ldots & 0 & 0\\
            H^{\otimes d} & 0 & \ldots & 0 & 0
        \end{bmatrix}.
        \label{eq:hmkmd2}
    \end{equation}
    The result follows from
    comparing the bottom left and upper right blocks of \eqref{eq:hmkmd} and \eqref{eq:hmkmd2}.
\end{proof}

\end{document}